\documentclass[authoryear,preprint,fleqn,12pt]{elsarticle}
\usepackage[utf8]{inputenc}
\usepackage{amsmath}
\usepackage{amsthm}
\newtheorem{theorem}{Theorem}
\newtheorem{lemma}[theorem]{Lemma}

\usepackage{bibunits}
\usepackage{graphicx}
\usepackage{amsfonts}
\usepackage[position=above]{subfig}
\usepackage{verbatim} 
\usepackage[toc,page]{appendix}

\usepackage[left=2cm,top=2cm,right=2cm]{geometry}
\usepackage{caption}
\newcommand*\patchAmsMathEnvironmentForLineno[1]{%
  \expandafter\let\csname old#1\expandafter\endcsname\csname #1\endcsname
  \expandafter\let\csname oldend#1\expandafter\endcsname\csname end#1\endcsname
  \renewenvironment{#1}%
     {\linenomath\csname old#1\endcsname}%
     {\csname oldend#1\endcsname\endlinenomath}}%
\newcommand*\patchBothAmsMathEnvironmentsForLineno[1]{%
  \patchAmsMathEnvironmentForLineno{#1}%
  \patchAmsMathEnvironmentForLineno{#1*}}%
\AtBeginDocument{%
\patchBothAmsMathEnvironmentsForLineno{equation}%
\patchBothAmsMathEnvironmentsForLineno{align}%
\patchBothAmsMathEnvironmentsForLineno{flalign}%
\patchBothAmsMathEnvironmentsForLineno{alignat}%
\patchBothAmsMathEnvironmentsForLineno{gather}%
\patchBothAmsMathEnvironmentsForLineno{multline}%
}
\usepackage{lineno}

\usepackage{setspace}
\renewcommand{\figurename}{Fig.}

\date{}

\journal{Theoretical Population Biology}

\begin{document}
\defaultbibliographystyle{model2-names}
\defaultbibliography{ecogenetic_main}

\begin{frontmatter}
\title{Ecological and genetic effects of introduced species on their native competitors}

\author{Meike J. Wittmann\corref{lab1}}
\cortext[lab1]{Corresponding author.}
\ead{wittmann@bio.lmu.de}

\author{Martin Hutzenthaler}
\author{Wilfried Gabriel}
\author{Dirk Metzler}

\address{Department Biologie II, Ludwig-Maximilians-Universit\"{a}t M\"{u}nchen, Großhaderner Str. 2, 82152 Planegg-Martinsried, Germany}

\begin{abstract}
Species introductions to new habitats can cause a decline in the population size of competing native species and consequently also in their genetic diversity. We are interested in why these adverse effects are weak in some cases whereas in others the native species declines to the point of extinction. While the introduction rate and the growth rate of the introduced species in the new environment clearly have a positive relationship with invasion success and impact, the influence of competition is poorly understood. Here, we investigate how the intensity of interspecific competition influences the persistence time of a native species in the face of repeated and ongoing introductions of the nonnative species.
We analyze two stochastic models: a model for the population dynamics of both species and a model that additionally includes the population genetics of the native species at a locus involved in its adaptation to a changing environment.
Counterintuitively, both models predict that the persistence time of the native species is lowest for an intermediate intensity of competition. This phenomenon results from the opposing effects of competition at different stages of the invasion process: With increasing competition intensity more introduction events are needed until a new species can establish, but increasing competition also speeds up the exclusion of the native species by an established nonnative competitor. By comparing the ecological and the eco-genetic model, we detect and quantify a synergistic feedback between ecological and genetic effects.
\end{abstract}


\end{frontmatter}

\pagebreak


\begin{bibunit}
\section{Introduction}
\noindent When a new species is introduced to a location where it did not occur before, it begins to interact with the resident species, for example as a predator, mutualist, or competitor. These interactions are critical in determining the fate of the introduced species and whether and how the community changes in response to the introduction. In this study, we investigate the effects of introduced species on the native species with which they are competing for resources, such as food or territories. Often there is considerable variation between geographical locations in the impacts of an introduced species on a particular native competitor, e.g.\ in container-dwelling mosquitoes in Florida \citep{3269Juliano1998} or for fish introductions to California \citep{3634Herbold1986}. Thus the question is: Why does the native species suffer from severe impacts or even goes extinct in some places, but not in others? 

While it is evident that the probability for an introduced species to become a high-impact invader at a certain location increases with introduction rate (also known as propagule pressure, one of the most important factors in invasion biology, \citealp{1054Duncan1997,1052Drake2005,350Lockwood2005}) and with the species' growth rate in the new environment \citep{3340Rejmanek1996,3345Wu2005}, the role of competition has remained controversial \citep{3634Herbold1986,1054Duncan1997,3305Davis2003}. Some empirical studies suggest that differences in the intensity of interspecific competition can explain differences in impacts, e.g.\ in the competitive interaction between native bumble bees and introduced honey bees in California, where the overlap in the flowers visited by honey bees and bumble bees was used as a proxy for the intensity of interspecific competition \citep{3280Thomson2006}. However, it has been difficult to disentangle the effect of competition from that of other variables and no consensus has emerged yet. Here we contribute a first theoretical building block towards an understanding of how competition intensity influences introduced species impacts.

Considerable impacts of an introduced species on a native species or even extinction of the native species can only occur if the introduced species completes two major stages of the invasion process: The establishment stage starts with the introduction of one or more founding individuals. This new population can then either go extinct, in which case it has to await the next introduction event, or increase in size until it is of the same order of magnitude as the competing native population. From there, the introduced species can proceed to a second stage in which it becomes dominant and may eventually exclude the native species from the community. Thus far, knowledge on the effect of competition is limited to single stages of the invasion process. 

At the establishment stage, the intensity of competition with native species appears to have a negative effect on the success of introduced species. With phylogenetic relatedness as a proxy for the intensity of competition, this intuitive idea goes back to \citet[][chapter 4]{98Darwin1859}. Darwin's naturalization hypothesis, as it is formulated nowadays, states that introduced species should be less successful at locations where closely related species are already present \citep{3203Duncan2002}. Surprisingly, the opposite pattern is sometimes observed in historical data, for example for plants introduced to Hawaii and New Zealand \citep{3201Daehler2001,3203Duncan2002}. A possible explanation is that closely related species not only compete with each other but also share characteristics that may confer a high intrinsic ability to survive and grow at the new location \citep{3203Duncan2002}. 

In the only experimental study on Darwin's naturalization hypothesis that we know of, however, the establishment success of an invader in a microcosm bacterial community increased with the average phylogenetic distance to the recipient community \citep{3311Jiang2010}. In this case, it has also been confirmed that more closely related species had a higher overlap in resource use and were therefore competing more intensely. Additional support for the negative effect of competition at the establishment stage comes from historical bird introductions to Hawaii: \citet{3248MOULTON1983} and \citet{3250MOULTON1993} observed a negative correlation between establishment success and the number of bird species already present, an observation that fits well with the predictions of a model for the assembly of a competitive Lotka-Volterra community \citep{3310Gamarra2005}.

At the second stage of the invasion process, when the nonnative species has already established, the impacts of the introduced species on the native competitors seem to increase with increasing intensity of competition. For example, among pairs of congeneric bird species introduced to Hawaii, pairs in which both species persisted had a significantly higher relative difference in beak length than pairs of species in which one or both species went extinct \citep{3198MOULTON1985}. Contrarily, \citet{3221Ricciardi2004} and \citet{3207Strauss2006} found that species less related to the native community had higher impacts. However, their measures of impact summarized effects on various aspects of the native community, not only those on competitors, such that their results do not fully apply to our problem. 

At this second stage, a number of evolutionary and genetic effects can contribute to the impacts of the introduced species on the native species. If the two species are closely related, there is the potential for hybridization and introgression, which can lead to new opportunities for evolution but also to extinction of rare native populations \citep{3270Rhymer1996}. Introduced and native competitors also impose new selection regimes on each other which can lead to shifts in life histories and resource use that reduce the intensity of competition and thus facilitate coexistence \citep{3268CROWDER1984}, or even a coevolutionary arms race in the exploitation of a limiting resource \citep{611Leger2010}. These evolutionary effects are strongly contingent on the species involved. 

There are also genetic effects, however, that are a direct consequence of the reduction in native population size due to competition and should therefore be present in most cases: smaller populations are subject to inbreeding depression, they can maintain a lower amount of genetic diversity, and accumulate deleterious mutations more rapidly \citep{1037Frankham1995,3052LANDE1995,3312Frankham2004}. Additionally, fluctuations in population size or habitat fragmentation can change the genetic configuration of the native species. So far, there are only a few empirical studies that examine such genetic effects of introduced species \citep[e.g.][]{1032KRUEGER1991,1071Kim2003}. In this study, we focus on the reduction in a native population's genetic diversity caused by a population decline after the invasion of a competitor. A reduction in genetic diversity can lower the native species' ability to respond to changes in the environment \citep{1040Strauss2006} and thus lead to a reduced growth rate \citep{1062Lande1996}. This, in turn, leads to a further decline in population size, thus closing the feedback loop. Such a synergistic feedback between ecological and genetic effects can accelerate population extinction \citep{3317Robert2011}. We will call this the eco-genetic effect of the introduced species and quantify how its strength depends on the intensity of competition with the native species.

In summary, a high intensity of competition between introduced and native species appears to have contrary effects at the different stages of the invasion process: it makes establishment of the introduced species more difficult but also increases the extinction risk imposed by an already established nonnative species on the native species \citep[see also][]{3355MacDougall2009}. This raises the question: What is the overall effect of competition intensity integrated across the entire invasion pathway?
 In this study, we consider a scenario in which an unlimited series of introduction events would sooner or later lead to the extinction of the native species. Our goal is to quantify for how long the native species can persist depending on the intensity of competition with the introduced species.

Stochastic models based on birth, death, and migration events at the individual level have increased our understanding of a wide range of processes in community ecology \citep{3654Black2012}, for example diversity patterns in dispersal-limited communities \citep{3187Alonso2006}. Stochasticity in the fates of individuals is particularly important for the dynamics of small introduced populations. Thus we use a stochastic modeling approach and compute the expectation and the variance of the time to the extinction of the native species, its persistence time. Very long persistence times in our model can be interpreted as signs for indefinite coexistence, since in these cases we would expect events like evolutionary divergence of niches to occur before the extinction of the native species \citep[see][and references therein]{1040Strauss2006}. 

We also study how the relationship between competition and persistence time is modulated by the rate at which nonnative individuals are introduced and the nonnative species' intrinsic ability to grow and reproduce in the new environment, i.e. we address questions such as: Is the effect of competition different for species that are introduced at a high vs. a low rate or does it depend on whether the introduced species or the native species has a higher fecundity. First, we consider these questions for a purely ecological model, which can be analyzed using theory on birth and death processes and a corresponding diffusion approximation. Then we transform this model into an eco-genetic model by adding a genetic dimension that allows us to quantify the feedback of reduced genetic diversity and adaptability on extinction risk in a variable environment and its dependence on competition strength. Lastly, we address a question of interest for invasive species management: How low does the introduction rate of a certain species need to be such that the native species is expected to persist for a certain threshold time? 

\section{Modeling}
\subsection{The ecological model}
\label{sec:ecomod}
We represent the population dynamics of the native and the introduced species as a continuous-time stochastic model similar to the Moran model in population genetics \citep{3178Moran1958}.
Consider a community consisting of a fixed number $K$ of individuals, each of which belongs either to the native or to the introduced species. The rate at which individuals die is proportional to the extent of competition experienced from conspecifics and members of the other species \citep[as in][]{3151Neuhauser1999}. The strength of interspecific competition relative to intraspecific competition is described by the non-negative competition coefficient $\alpha$. Small values of $\alpha$ represent weak interspecific competition and high values intense interspecific competition. In principle, this parameter can be estimated from data on the overlap in resource use between the two species \citep{3284MAY1975}. The individual whose offspring replaces the dead individual is selected by randomly drawing one individual from the whole community, including the individual that just died. In this draw, native individuals have weight 1 and members of the introduced species weight $w$. Thus $w$ can be understood as the fecundity of the introduced species relative to that of the native species in the sense that it is proportional to the per capita number of offspring in a large offspring pool from which the new individual is drawn. 

Note that for the sake of simplicity we assume that the competition coefficient $\alpha$ is the same for both species. Thus we are using this parameter to describe the symmetric component of the competitive interaction. Nevertheless, an asymmetric situation can be generated by setting the fecundity parameter $w$ to a value different from 1. For $w>1$ the introduced species has a fecundity advantage, for $w<1$ a disadvantage. Initially, all $K$ individuals belong to the native species. From time 0 onwards, single individuals of a nonnative species are introduced at rate $\gamma$ (on average $\gamma$ times per time unit) and start competing with the native species. To bring the community back from $K+1$ to $K$ individuals after an introduction event, one individual is drawn to die with weights proportional to the competition experienced.

The population dynamics of the native species can be formulated as a Markov process $N=\left(N(t)\right)_{t\geq 0}$ with state space $\{0,1,2,\dots,K\}$ which describes the number of native individuals currently in the community. Since in this model transitions are only possible between neighboring states, it belongs to the class of birth and death processes \citep[][p.131-150]{3100Karlin1975}. The rate $\lambda_n$ at which the number of native individuals increases by one is
\begin{linenomath}\begin{equation}
\lambda_n =\underbrace{\frac{c(K-n,n)\cdot (K-n)}{K}}_{\substack{\text{rate at which}\\ \text{members of the introduced}\\ \text{species die}}} \cdot \underbrace{\frac{n}{(K-n)\cdot w+ n}}_{\substack{\text{probability that a}\\ \text{native individual} \\ \text{gives birth}}}
\label{eq:lambda_def}
\end{equation}\end{linenomath}
for $n \in \{0,\dots,K-1\}$, where $c(x,y)=x+\alpha y$ is the competition experienced by an individual when the population size of the species it belongs to is $x$ and the size of the competing species is $y$. The constant of proportionality for the death rate is chosen to be $1/K$, such that in the absence of the introduced species native individuals die at rate 1. Then there are on average $K$ death events per time unit and one time unit can be considered as one generation. 
For $n \in \{1,\dots,K\}$, the rate $\mu_n$ at which the number of native individuals decreases by one if there are currently $n$ native individuals is
\begin{linenomath}\begin{align}
\mu_n = & \underbrace{\frac{c(n,K-n)\cdot n}{K}}_{\substack{\text{rate at which}\\ \text{native individuals die}}} \cdot \underbrace{\frac{(K-n)\cdot w}{(K-n) \cdot w+ n}}_{\substack{\text{probability that an}\\ \text{introduced individual} \\ \text{gives birth}}} \nonumber \\
 & + \underbrace{\gamma}_{\substack{\text{introduction} \\ \text{rate}}} \cdot \underbrace{\frac{c(n,K-n+1)\cdot n}{c(n,K-n+1)\cdot n + c(K-n+1,n)\cdot(K-n+1)}}_{\substack{\text{probability that a native}\\ \text{individual dies}}}\:.
\label{eq:mu_def}
\end{align}\end{linenomath}

The assumption of a fixed community size is a good approximation for pairs of ecologically similar species for which interspecific competition is as strong as intraspecific competition for the resource which is limiting population size. This can happen, for example, if there is a fixed number of territories or nesting places that can be occupied by one individual from either species. The competition for other important resources can be less intense ($\alpha<1$) or there can be interspecific interference ($\alpha>1$). The robustness of our model results against violations of the constant community size assumption is explored in the supplementary material. 

In a community of finite size, coexistence of native and introduced species is not possible in the long run. We assume that our model encompasses the whole range of the native species, such that a reintroduction of native individuals from outside is not possible. The introduced species, in contrast, can fail to establish and go extinct after an introduction event, but will then be reintroduced at a later time. Therefore, the only absorbing state of the model is 0, the state at which the whole community consists of introduced individuals and the native species is extinct. Note that in the absence of immigration, the native species would be able to persist for an infinite amount of time, since by assumption the total number of individuals in the community is constant. 

In the symmetric case ($w=1$), the rare species has an advantage over the more common species for $\alpha < 1$ because then $c(n,K-n)=n + \alpha (K-n) < K-n + \alpha n=c(K-n,n)$ if $n < K/2$. This leads to fluctuations of the system around a point (which we will call the coexistence point) in which each species has population size $K/2$. For asymmetric competition ($w \neq 1$), whether coexistence is possible on an intermediate time scale, depends on the values of $\alpha$ and $w$. Also the position of the coexistence point depends on these parameters. 

Let $T_n$ be the random time to extinction of the native species in a realization of the process that starts with $n$ native individuals. Let $\tau_n$ and $\sigma_n^2$ denote the expected value and the variance of $T_n$. We will use the ecological model to compute the expected value $\tau_K$ and the variance $\sigma_K^2$ of the time to the extinction of the native species when it is starting with population size $K$, i.e.\ in the state in which the nonnative species is still absent.

\subsection{The eco-genetic model}
\label{sec:ecogenmod}
Now we extend the ecological model by including a genetic component. To keep the model tractable, we chose the simplest possible genetic scenario: We assume that each native individual is haploid and possesses one bi-allelic locus which determines the individual's response to some environmental factor, for example whether or not the individual is resistant to a certain parasite. At any point in time, one of the two alleles is favored and its carriers have fecundity 1, whereas other native individuals have fecundity $1-s$. Thus $s$ is a measure for the strength of selection. With probability $u$ an offspring mutates to the respective other allele and at rate $\epsilon$ the environment and with it the currently favored allele changes.

We assume that before introductions start, the number of native individuals that carry the favored allele has reached a stationary distribution, i.e.\ it is in mutation-selection equilibrium (see section \ref{sec:ecogenetic_full} for a derivation of this stationary distribution). To be able to compare the results of the eco-genetic model to those of the ecological model with the same introduced species fecundity parameter $w$, we multiplied $w$ in the eco-genetic model by the average fecundity $w^*$ (see \eqref{eq:w*}) of native individuals under the stationary distribution. Thus in both models, $w$ can be interpreted as the fecundity of the introduced species relative to the average fecundity of native individuals. 

This model can be represented as a Markov process, where the state with $n$ native individuals, $m$ of which carry the currently favored allele, is denoted by $(n,m)$. From this state, we can reach the states $(n+1,m),(n+1,m+1),(n,m+1),(n,m-1),(n-1,m-1)$, and $(n-1,m)$ through birth-death events, which possibly involve a mutation in the first four cases. The two latter states can also be reached through introduction events. The transition rates are defined analogously to those in the ecological model, with the death rate proportional to the competition experienced and the probability of giving birth proportional to fecundity. The main difference to the ecological model is that now the native population is divided into two allelic classes between which individuals can switch by mutation. 

As an example (see \ref{sec:ecogenetic_full} for all other transition rates), the transition rate from state $(n,m)$ to state $(n,m+1)$ is
\begin{equation}
\frac{c(n,K-n)\cdot (n-m)}{K}\cdot \frac{m \cdot (1-u)+(1-s)(n-m)\cdot u}{m+(1-s)(n-m)+w^* \cdot w \cdot (K-n)}\:.
\end{equation}
This is the rate at which one of the $n-m$ native individuals carrying the disfavored allele dies multiplied by the probability that it is replaced by a native individual with the favored allele. This new individual can either be the offspring of one of the $m$ individuals with a favored allele that did not mutate (probability $1-u$) or an offspring of one of the $n-m$ parents with the disfavored allele that mutated (probability $u$). The transition from $(n,m)$ to $(n,n-m)$ represents a change of the environment and happens at rate $\epsilon$.

Our goal here is to compute the expected time to extinction of the native species. Here we start in the state $(K,m^*)$, where the introduced species is absent and $m^*$ is the average number of native individuals that carry the favored allele under the stationary distribution rounded to the next integer.

\section{Results}

\subsection{Ecological effect}
Measured from the start of introductions when the native population has still size $K$, the persistence time of the native species has expectation
\begin{equation}
\tau_K = \mathbf{E}[T_K]=\sum_{l=1}^{K}\sum_{i=l}^{K}\frac{1}{\mu_i} \prod_{j=l}^{i-1} \frac{\lambda_j}{\mu_j}
\label{eq:T_K_solution}
\end{equation}
and its variance is
\begin{equation}
\sigma_K^2= \mathbf{Var}(T_K)=\sum_{l=1}^{K}\sum_{i=l}^{K} \eta_i \prod_{j=l}^{i-1} \frac{\lambda_j}{\mu_j}\:,
\label{eq:variancesolution}
\end{equation}
where
\begin{equation}
\eta_i = \begin{cases}
\frac{1}{\mu_i(\mu_i+\lambda_i)}\left[1+\mu_i \lambda_i \left(\tau_{i+1}-\tau_{i-1}\right)^2\right] & \text{if }1\leq i \leq K-1 \\
\frac{1}{\mu_K^2} & \text{if } i=K
\end{cases}\:.
\label{eq:etas}
\end{equation}

The results \eqref{eq:T_K_solution} and \eqref{eq:variancesolution} are derived by noticing that the time to the extinction of the native species when starting in a state $n \in \{1,\dots,K\}$ has the same distribution (denoted $\overset{d}{=}$) as the sum of two independent random variables: 
\begin{equation}
T_{n} \overset{d}{=} S_{n} + T_{N'} \overset{d}{=} S_{n} + \begin{cases}
T_{n+1} & \text{ with probability } \frac{\lambda_n}{\lambda_n + \mu_n} \\
T_{n-1} & \text{ with probability } \frac{\mu_n}{\lambda_n + \mu_n}
\end{cases}
\label{eq:decomposition}
\end{equation}
for $n \in \{1,\dots,K-1\}$ and $T_{K} \overset{d}{=} S_{K} + T_{K-1}$. Here, $S_{n}$ is a random variable for the time until the native population size first changes from state $n$ to a new state $N'$, which in our model is either $n-1$ or $n+1$. $S_{n}$ is exponentially distributed with parameter $\lambda_n + \mu_n$ for $n \in \{1,\dots,K-1\}$ and with parameter $\mu_K$ for $n=K$. Taking the expectation on both sides in \eqref{eq:decomposition} and using $T_{0}=0$ we obtain the following recursion for $\tau_n$ \citep[see][p. 145-150 for similar problems and solutions]{3100Karlin1975}:
\begin{equation}
\tau_n = \begin{cases}
0 & \text{if } n=0 \\[0.5em]
\frac{1}{\lambda_n + \mu_n} + \frac{\lambda_n}{\lambda_n + \mu_n}\tau_{n+1} + \frac{\mu_n}{\lambda_n + \mu_n}\tau_{n-1} & \text{if }1\leq n \leq K-1 \\[0.5em]
\frac{1}{\mu_K}+\tau_{K-1} & \text{if } n=K
\end{cases}\:,
\label{eq:recursion_full}
\end{equation}
which can be solved for $\tau_K$. 

Using the law of total variance 
\begin{equation}
\mathbf{Var}(T_{N'})= \mathbf{E}\left[\mathbf{Var}(T_{N'}|N')\right]+\mathbf{Var}\!\left(\mathbf{E}[T_{N'}|N']\right)=\mathbf{E}[\sigma_{N'}^2]+\mathbf{E}[\tau_{N'}^2]-\mathbf{E}[\tau_{N'}]^2\:,
\end{equation}
we obtain a similar recursion for the variance
\begin{equation}
\sigma_n^2 = \begin{cases}
0 & \text{if } n=0 \\[0.5em]
\frac{1}{(\lambda_n + \mu_n)^2} + \frac{\lambda_n}{\lambda_n + \mu_n}\left(\sigma_{n+1}^2+\tau_{n+1}^2\right) + \frac{\mu_n}{\lambda_n + \mu_n}\left(\sigma_{n-1}^2+\tau_{n-1}^2\right) & 
\\- \left(\frac{\lambda_n}{\lambda_n + \mu_n}\tau_{n+1} + \frac{\mu_n}{\lambda_n + \mu_n} \tau_{n-1}\right)^2 & \text{if }1\leq n \leq K-1 \\[0.5em]
\frac{1}{\mu_K^2}+\sigma_{K-1}^2 & \text{if } n=K
\end{cases}\:,
\label{eq:recursion_variance}
\end{equation}
which can be solved analogously to \eqref{eq:recursion_full} once the $\tau_n$ are known (see \ref{sec:recursionsolution} for details of these derivations).

The result given by \eqref{eq:T_K_solution} reveals that with increasing strength of competition, the expected time to extinction $\tau_K$ decreases until interspecific and intraspecific competition are of similar strength ($\alpha \approx 1$) (Fig. \ref{fig:omega0} A). Here, $\tau_K$ reaches a minimum. If the strength of interspecific competition is further increased, $\tau_K$ grows again. The minimizing competition coefficient is below one for low introduction rates, and above one for large introduction rates (Fig. \ref{fig:omega0} B).
For low introduction rates ($\gamma<1$), the minimum moves towards lower competition coefficients, i.e. weaker competition, if fecundities are unequal, no matter whether the introduced species has a higher ($w > 1$) or a lower fecundity parameter ($w <1$) than the native species (Fig. \ref{fig:delta_effect} A). As expected, $\tau_K$ decreases with increasing introduction rate and increasing fecundity advantage of the introduced species.

These patterns in $\tau_K$ are paralleled by a corresponding pattern in the variance of the expected time to extinction $\sigma_K^2$. The variance increases with increasing $\tau_K$ and thus also exhibits a minimum. To compare the distribution of extinction times to an exponential distribution where the expected value equals the standard deviation, we computed the ratio between the standard deviation and the expected value of the persistence time in our model. This ratio is close to one for parameter combinations that lead to a high persistence time and below one for parameter combinations where the extinction of the native species is relatively fast (Fig. \ref{fig:varianceratio}).

\begin{figure}
  \centering
  \subfloat[]{\label{fig:omega0macro}\includegraphics[width=0.5\textwidth,trim = 0mm 0mm 0mm 15mm, clip]{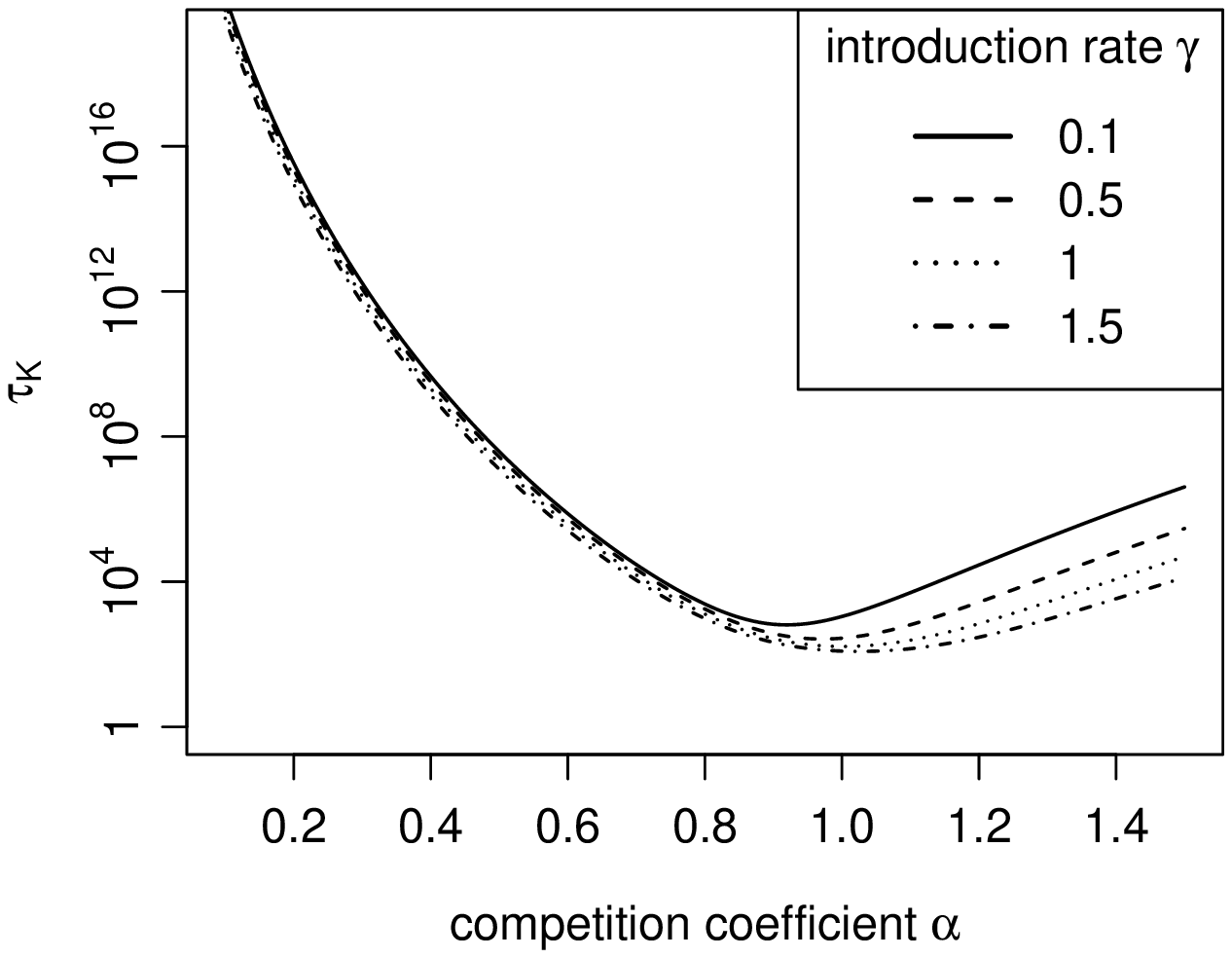}}                
  \subfloat[]{\label{fig:omega0zoom}\includegraphics[width=0.5\textwidth,trim = 0mm 0mm 0mm 15mm, clip]{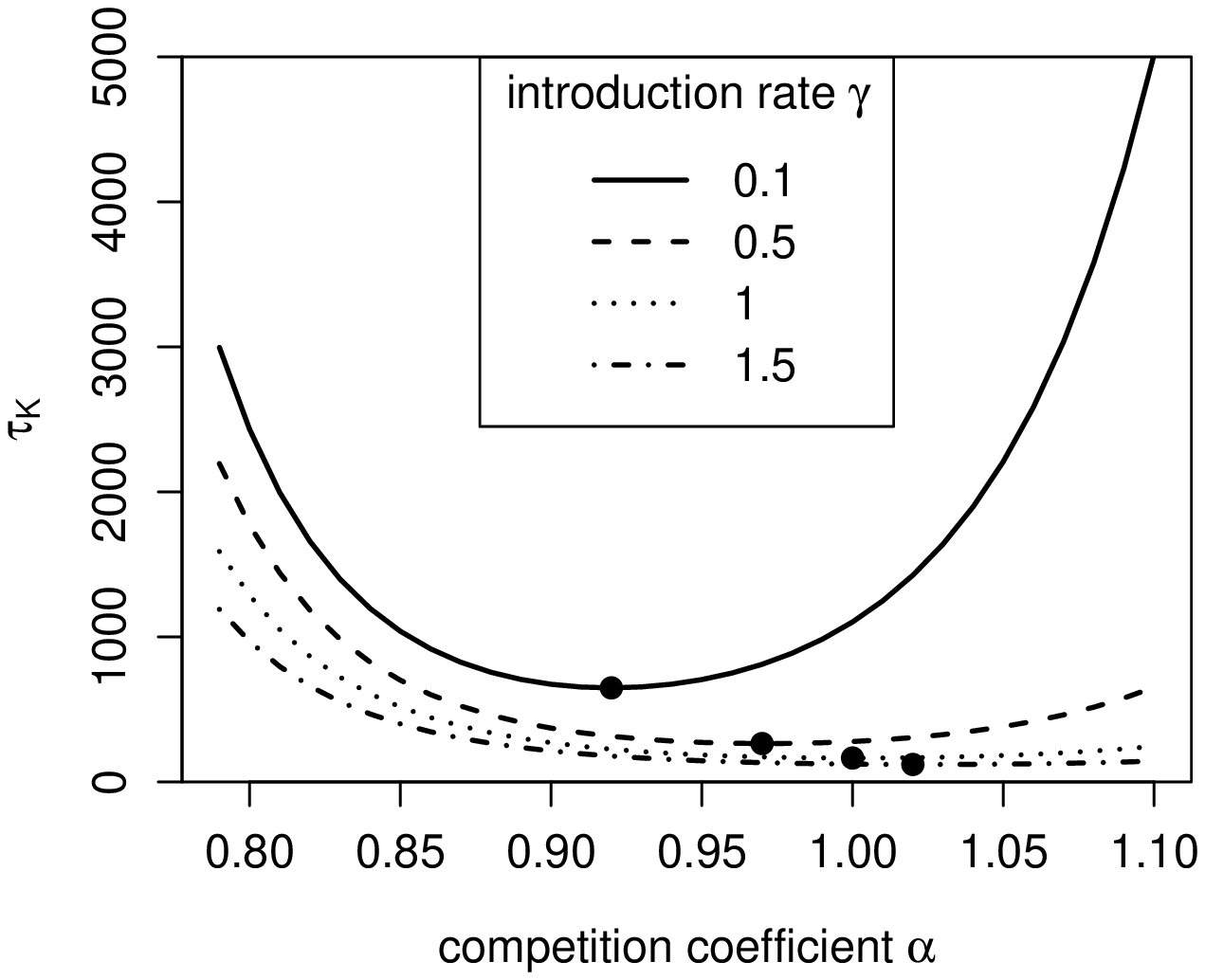}}
  \caption{The expected time to extinction $\tau_K$ for different values of the introduction rate $\gamma$ as a function of the competition coefficient $\alpha$. (B) magnifies the part of (A) around $\alpha=1$. The minimum of each curve in B is indicated by a solid point. ($K$ = 100, $w=1$.) }
  \label{fig:omega0}
\end{figure}
\begin{figure}
  \centering
  \includegraphics[height=.8\textheight,trim = 0mm 0mm 0mm 5mm, clip]{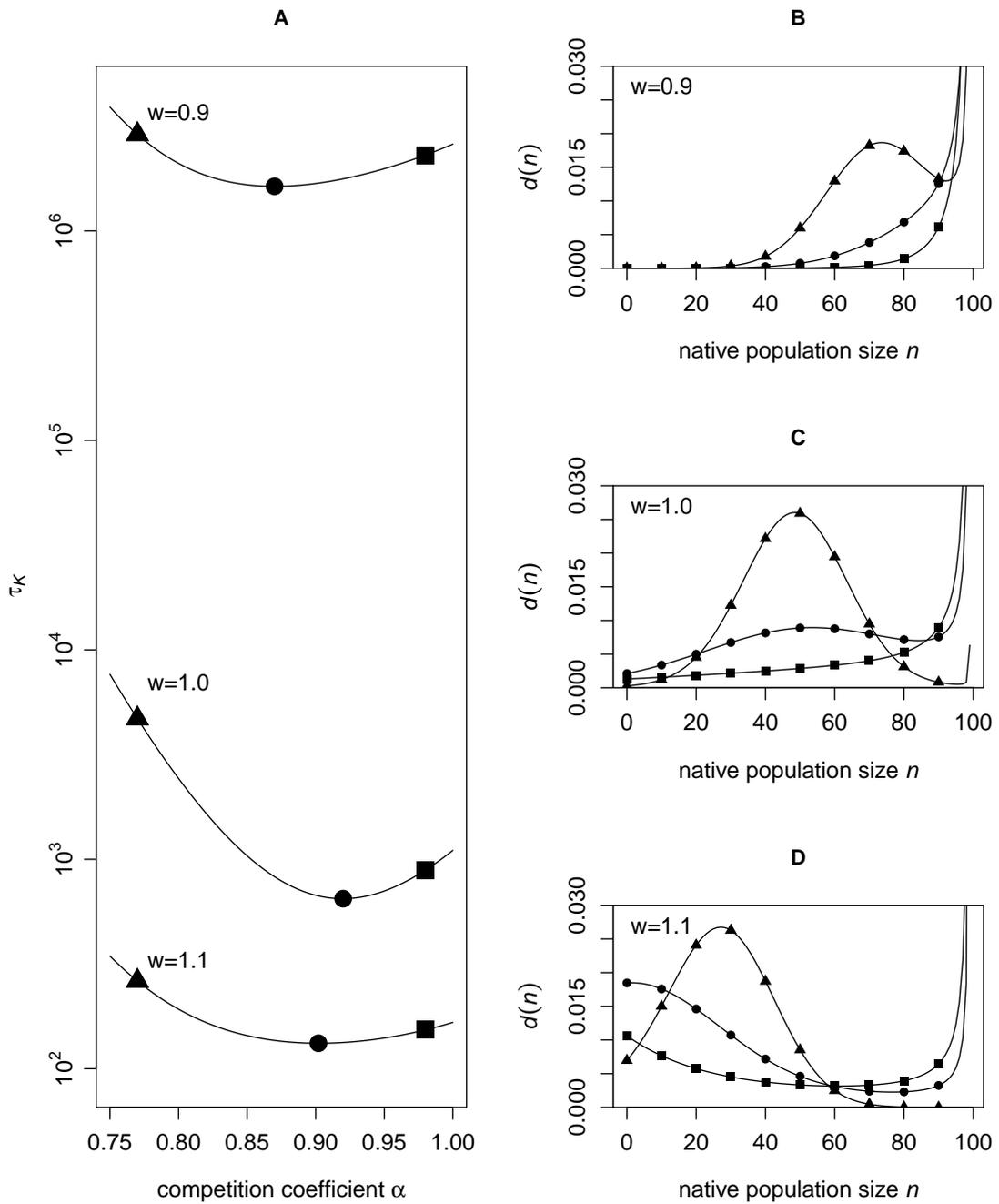} 
  \caption{The effect of changes in the introduced species fecundity parameter $w$ on the expected time to the extinction of the native species and the position of the minimum (A) and on the quasi-stationary distribution $(d(1),\dots,d(K))$ of the Markov process  conditional on non-extinction of the native species (right column). The curves in the right column correspond to the competition coefficients marked with the respective symbol in A. ($K=100$, $\gamma=0.1$.)}
  \label{fig:delta_effect}
\end{figure}
\begin{figure}
\centering
 \subfloat[]{\includegraphics[width=0.5\textwidth,trim = 0mm 0mm 0mm 15mm, clip]{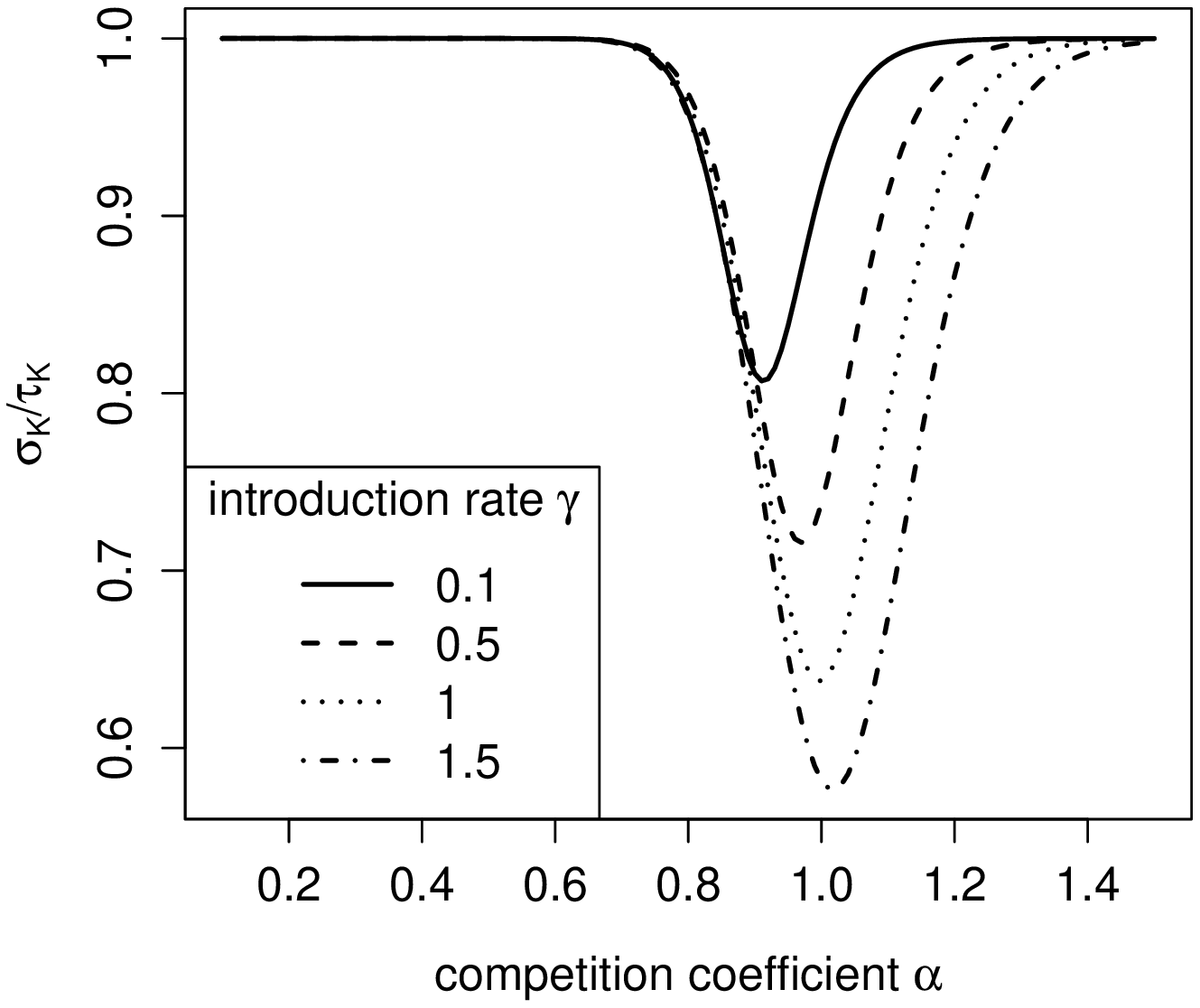}}
 \subfloat[]{\includegraphics[width=0.5\textwidth,trim = 0mm 0mm 0mm 15mm, clip]{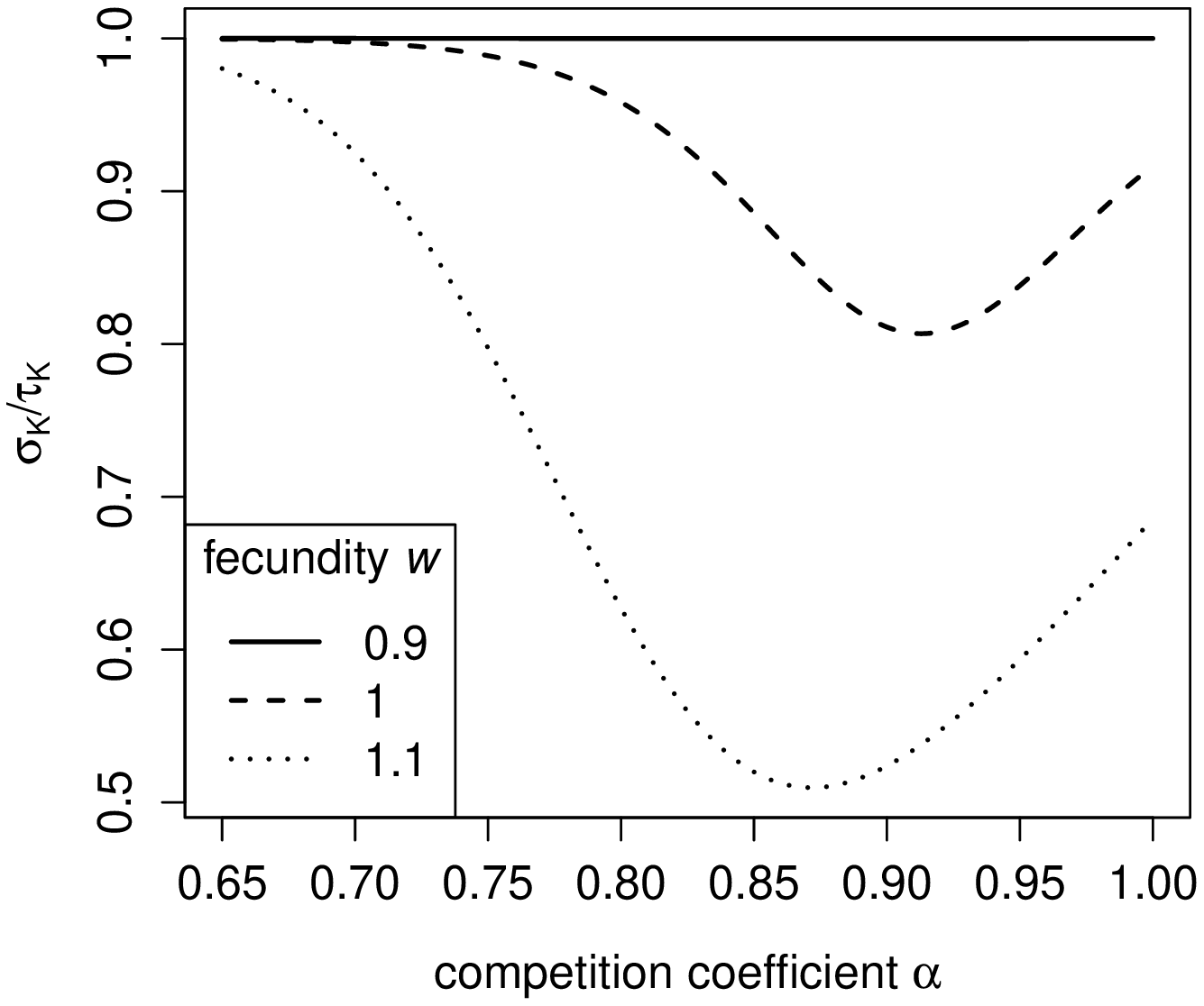}}
\caption{The ratio of the standard deviation and the expected value of the time to extinction for (A) different introduction rates with $w=1$ and (B) for different values of the introduced species fecundity with $\gamma=0.1$. ($K=100$.)}
\label{fig:varianceratio}
\end{figure}

\vspace{3mm}
The numerical evaluation of \eqref{eq:T_K_solution} is practical only for small community sizes. For moderate to large community sizes we derived an approximation for the expected time to extinction that is easier to compute and gives more insight into the dependence of persistence time on the parameters:
\begin{equation}
\tau_K \approx K \cdot \int_0^1 \frac{1}{(1-\xi)^{\gamma}} \int_\xi^{1} \frac{(1-\eta)^{\gamma-1}}{\eta}\cdot e^{\beta\left[\xi(1-\xi)-\eta(1-\eta)\right]} \cdot e^{\delta (\xi-\eta)} d\eta \: d\xi,
\label{eq:T_Kdiff}
\end{equation}
where $\beta := (\alpha-1)\cdot K$ is the rescaled advantage of being common and $\delta := (w - 1)\cdot K$ is the rescaled fecundity advantage of the introduced species.
\vspace{3mm}

The result \eqref{eq:T_Kdiff} is based on a diffusion approximation of the birth and death process described by equations \eqref{eq:lambda_def} and \eqref{eq:mu_def}. The process is rescaled such that the native population size is expressed as a fraction of the total community and time is sped up by a factor $K$:
\begin{equation}
X_K = \left(X_K(t)\right)_{t\geq0}=\left(\frac{N(K\cdot t)}{K}\right)_{t \geq 0}.
\label{eq:rescaledprocess}
\end{equation}
In the limit as $K$ goes to infinity while $\beta$ and $\delta$ are held constant, $X_K$ converges in distribution to the diffusion process $X$ with infinitesimal generator \citep[see][p. 195, and the supplementary material for a derivation]{3238Karlin1981}: 
\begin{linenomath}\begin{align}
Lf(x) & :=\frac{d}{dt} E\big[f\big(X\left(t\right)\big)|X(0)=x\big]\big{|}_{t=0} \nonumber \\
 & = x(1-x)\frac{d^2}{dx^2}f(x) + \big(-\beta \left(1-2x\right)\left(1-x\right)x - \gamma x - \delta x (1-x)\big)\frac{d}{dx}f(x)\:.
\label{eq:generator}
\end{align}\end{linenomath}

The expected time to extinction of the native species when its starting frequency is $x$ is a solution $g(x)$ of the differential equation \citep[][p. 193]{3238Karlin1981}
\begin{equation}
Lg(x)=-1\:
\label{eq:differentialequation}
\end{equation}
with boundary conditions $g(0)=0$ and $|\lim\limits_{x \nearrow 1} g'(x)|<\infty$, where $\lim\limits_{x \nearrow 1} g'(x)$ denotes the left-sided limit at 1 (see supplementary material for details).

By numerically evaluating \eqref{eq:T_Kdiff} in R \citep{575Team2009} and using a golden section search algorithm \citep{1227Heath2002} we computed $r(\gamma,\delta)$, the value of $\beta$ that minimizes the right hand side of \eqref{eq:T_Kdiff} for given values of $\gamma$ and $\delta$. After rescaling, we obtained for the competition coefficient $\hat{\alpha}$ which minimizes the expected time to the extinction of the native species.
\begin{equation}
\hat{\alpha} = 1+ \frac{r(\gamma,\delta)}{K}\:.
\end{equation}
With increasing introduction rate $\gamma$, $r(\gamma,\delta)$ increases (Fig. \ref{fig:rgammadelta} A), becoming positive at $\gamma=1$. The absolute value of $r(\gamma,\delta)$ increases with the differences in fecundity between the species (Fig. \ref{fig:rgammadelta} B).

\begin{figure}
  \centering           
  \subfloat[]{\label{fig:r(delta)}\includegraphics[width=0.5\textwidth,trim = 0mm 0mm 0mm 15mm, clip]{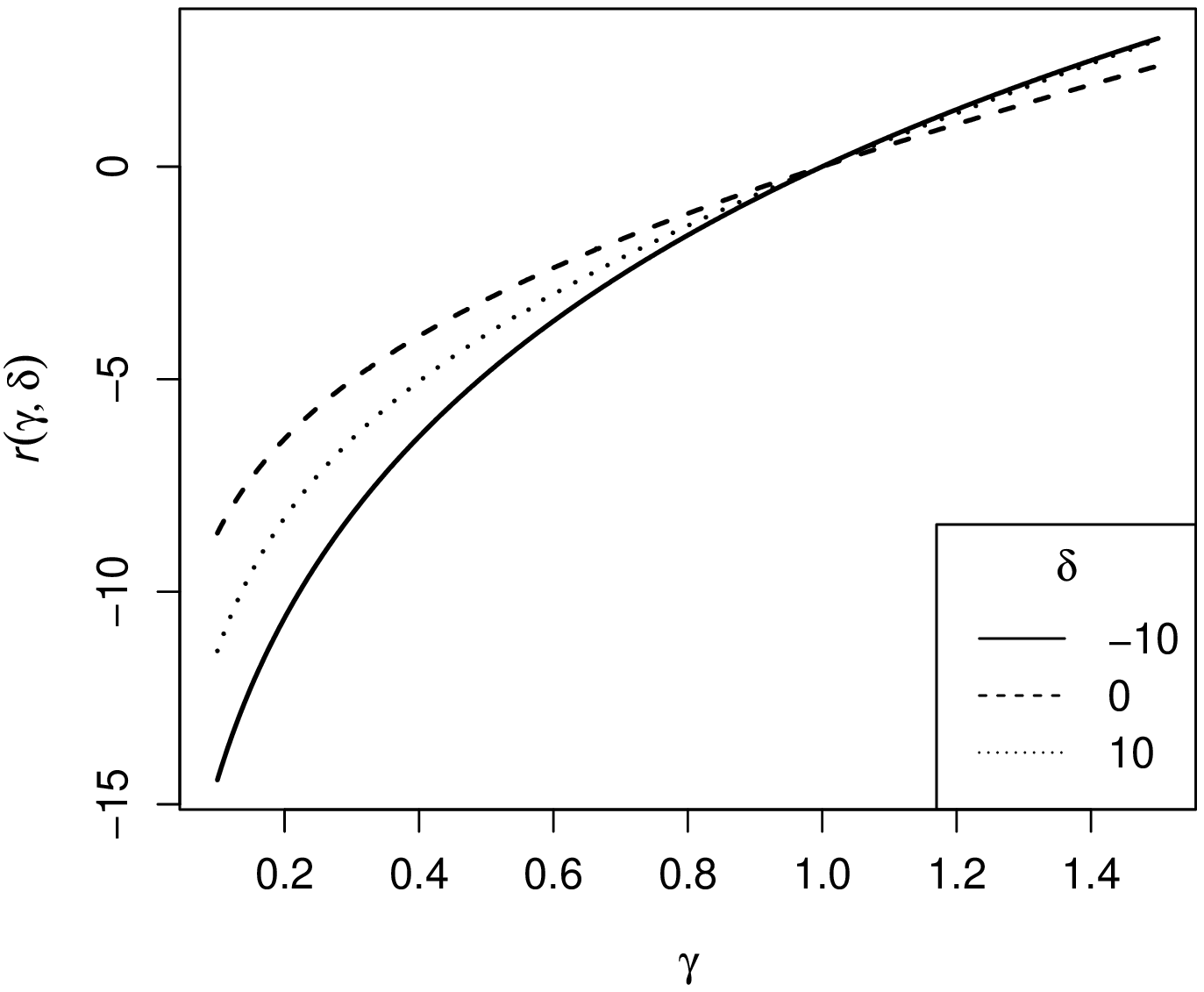}}
  \subfloat[]{\label{fig:r(gamma)}\includegraphics[width=0.5\textwidth,trim = 0mm 0mm 0mm 15mm, clip]{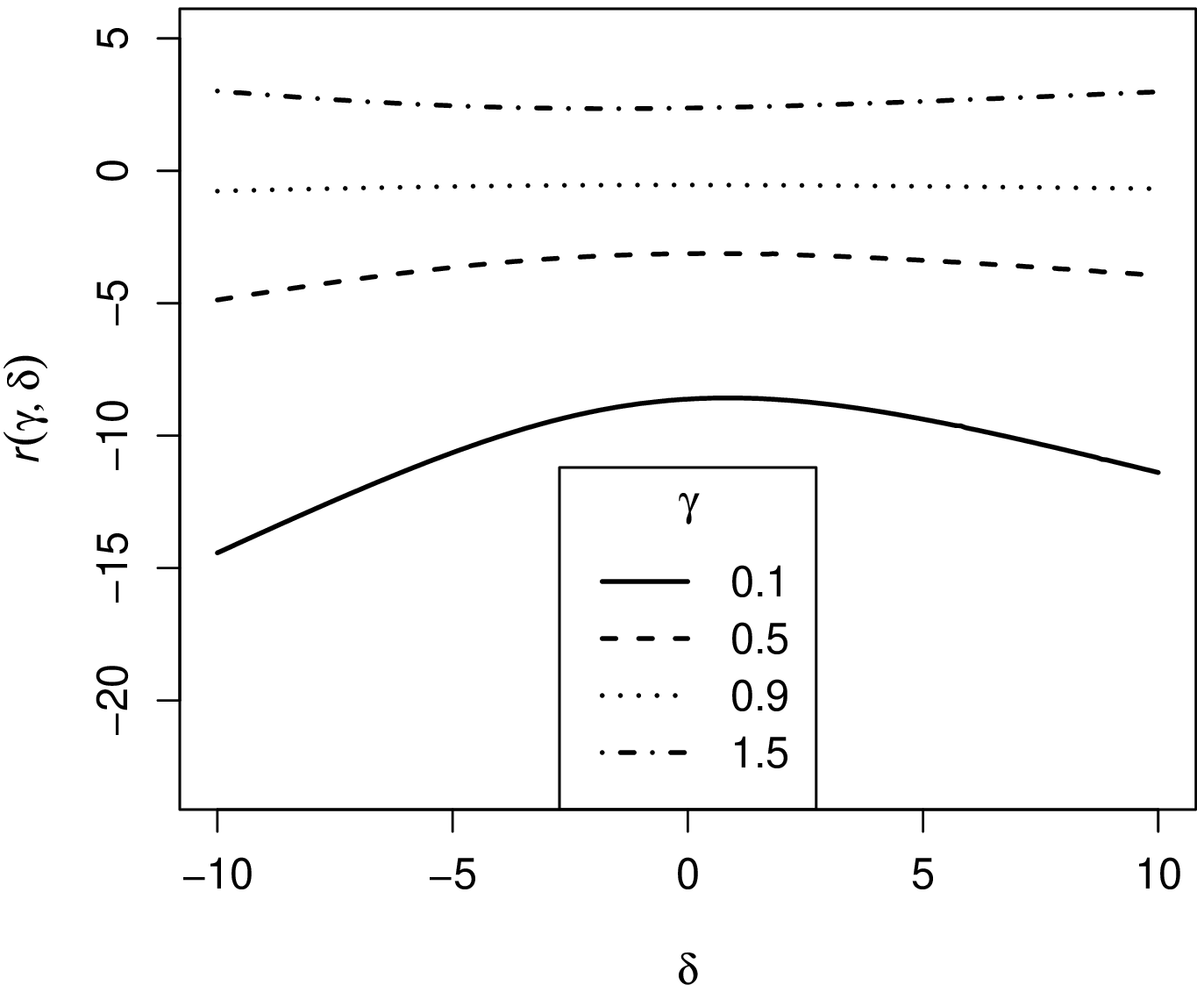}}     
  \caption{The minimizing rescaled competition coefficient $r(\gamma,\delta)$ as a function of the introduction rate $\gamma$ with fixed rescaled fecundity parameter $\delta$ (A) and as a function of $\delta$ with fixed $\gamma$ (B).}
  \label{fig:rgammadelta}
\end{figure}

\subsection{Eco-genetic effect}
As in \eqref{eq:recursion_full}, we recursively computed $\tau_{(n,m)}$, the expected time to the extinction of the native species when starting in state $(n,m)$, by decomposing it according to what happens at the first jump. Doing this for all states gave rise to a system of $\frac{(K+1)(K+2)}{2} -1$ linear equations, which we solved numerically in R for $\tau_{(K,m^*)}$, where $m^*$ is the average initial number of native individuals with the favored allele.

Although we adjusted the fecundity of the introduced species to match the average fecundity of the native species, the expected time to extinction of the native species is lower under the eco-genetic model than under the ecological model described in Section \ref{sec:ecomod} (Fig. \ref{fig:ecoevo}). Over wide regions in parameter space, $\tau_{(K,m^*)}$ decreases with increasing selection strength $s$ acting on the native species (Fig. \ref{fig:ecoevo} A) and increases with increasing mutation probability $u$ (Fig. \ref{fig:ecoevo} B). Fig. \ref{fig:ecoevo} C indicates that in the absence of environmental change, the expected time to extinction is similar to its counterpart in the ecological model. For a non-zero rate of change, the expected time to extinction is reduced, with a particularly strong reduction at small competition coefficients. In all cases, the minimizing competition coefficient (indicated by solid points in Fig. \ref{fig:ecoevo}) is reduced if we take into account the eco-genetic feedback.

\begin{figure}
  \centering
  \subfloat[]{\label{fig:ecogenetic_s}\includegraphics[width=0.35\textwidth,trim = 0mm 0mm 0mm 10mm, clip]{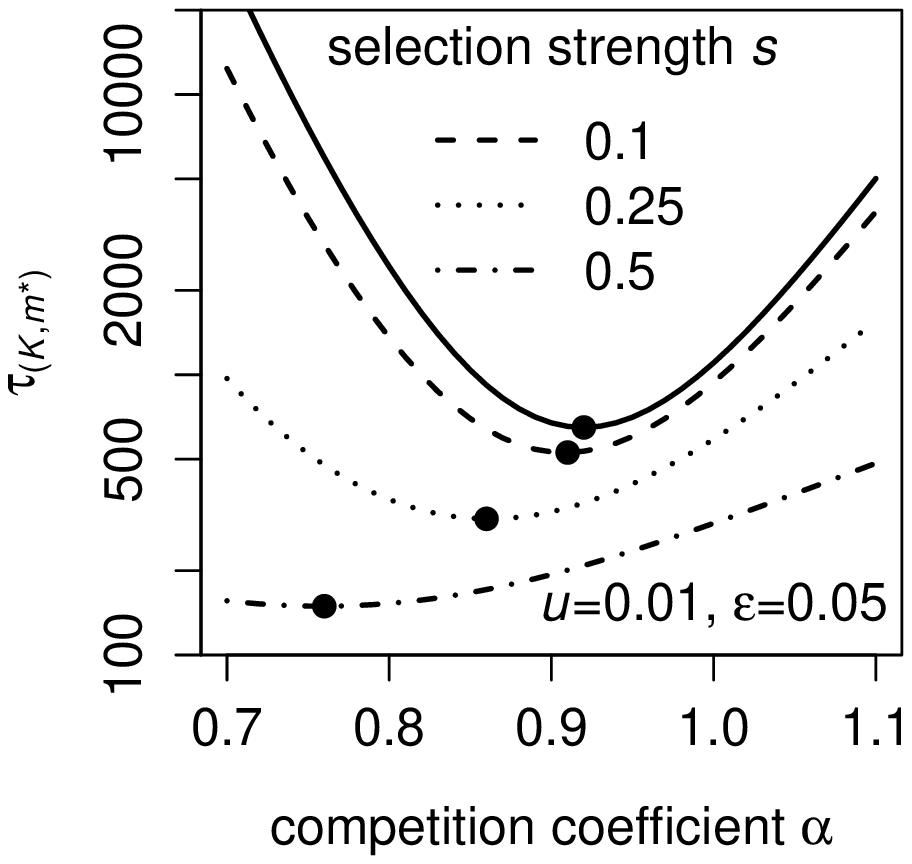}}
  \subfloat[]{\label{fig:ecogenetic_u}\includegraphics[width=0.35\textwidth,trim = 0mm 0mm 0mm 10mm, clip]{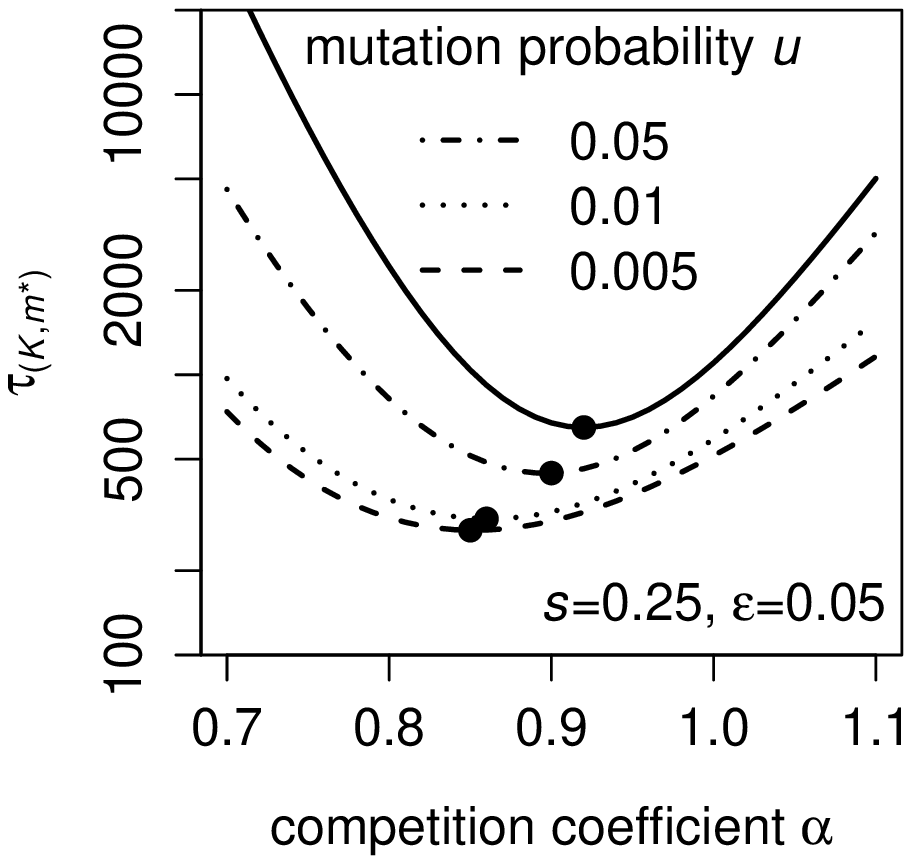}}
  \subfloat[]{\label{fig:ecogenetic_epsilon}\includegraphics[width=0.35\textwidth,trim = 0mm 0mm 0mm 10mm, clip]{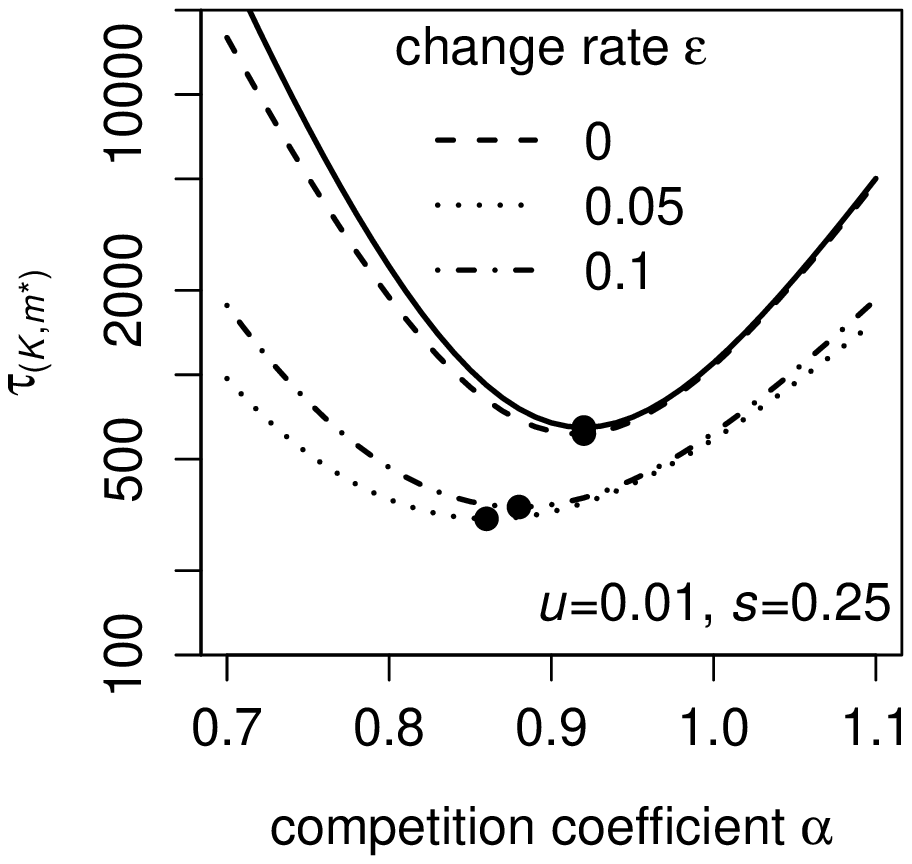}}
  \caption{The expected time to the extinction of the native species under the eco-genetic model as a function of the competition coefficient for different selection coefficients (A), different mutation probabilities (B), and different rates of environmental change (C). The solid line corresponds to the expected time under the ecological model (Eq. \eqref{eq:T_K_solution}). Minima are indicated by solid points. ($K=100, \gamma=0.1, w=1$.)}
  \label{fig:ecoevo}
\end{figure}

\subsection{Critical introduction rate}
We set a threshold persistence time $\tau_{crit}$ and determined the critical introduction rate $\gamma_{crit}$, such that $\tau_K > \tau_{crit}$ for all $\gamma < \gamma_{crit}$, using a bisection algorithm \citep{1227Heath2002} for both the ecological and the eco-genetic model. To be able to compare the two models, we adjusted the fecundity parameter of the introduced species as above to match the average fecundity of the native species under mutation-selection equilibrium in a population of size $K$.
As was the case with the expected time to extinction for fixed introduction rate, the critical introduction rate also reaches a minimum at an intermediate competition coefficient (Fig. \ref{fig:gammacrit}). Not surprisingly, the critical introduction rate decreases with increasing fecundity advantage of the introduced species. Taking the ecogenetic effect into account, the critical introduction rate is lower than under the purely ecological model.
\begin{figure}
  \centering
  \includegraphics[width=0.5\textwidth]{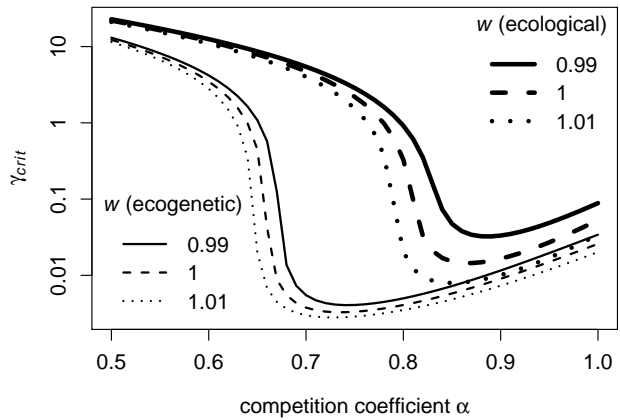}            
  \caption{The critical introduction rate (on a logarithmic scale) for different fecundity parameters $w$ under the ecological and eco-genetic model as a function of the competition strength. In the eco-genetic case, $w$ was multiplied by the average fecundity of native individuals under mutation-selection balance. ($\tau_{crit}=2000$, $\epsilon=0.05$, $s=0.25$, $u=0.01$.)}
  \label{fig:gammacrit}
\end{figure}

\section{Simplifications and heuristics}
To intuitively understand why competitors with intermediate interaction coefficients lead to the lowest persistence time of the native species, we simplify the state space to three possible states (Fig. \ref{fig:simplemodel} A): the introduced species is absent and the whole community consists of the native species ($N$), both native and introduced species coexist ($N\&I$), or the native species is extinct and the whole community consists of the introduced species ($I$). If both species have the same fecundity ($w = 1$) and the introduction rate is small, native and introduced species exclude each other with approximately the same probability from the coexistence point and the rate $\phi(\alpha)$ at which this happens depends very little on $\gamma$. Let $\psi(\alpha,\gamma)$ be the rate at which the nonnative species establishes. Then, analogously to \eqref{eq:recursion_full}, we can formulate a recursion for the expected time to extinction of the native species:
\begin{equation}
\tau_N = \frac{1}{\psi(\alpha,\gamma)}+\tau_{N\&I}
\end{equation}
and
\begin{equation}
\tau_{N\&I}=\frac{1}{2\phi(\alpha)}+\frac{1}{2}\tau_N\:.
\end{equation}

As solution for the expected time to extinction of the native species when there is currently no introduced individuals we obtain
\begin{equation}
\tau_N = \frac{2}{\psi(\alpha,\gamma)}+\frac{1}{\phi(\alpha)}\:.
\end{equation}
This is essentially twice the sum of the expected sojourn times in the states $N$ and $N\&I$. What are these times?
In the full model, $1/\psi(\alpha,\gamma)$ approximately corresponds to the expected time for the introduced species to reach population size $K/2$ starting from size 0. We will refer to this time as the establishment time of the introduced species. The term $1/(2\phi(\alpha))$ is the expected time for one of the two species to go extinct when they are currently coexisting with population size $K/2$ each. We will call this the exclusion time. Expressions for establishment time and exclusion time as functions of the model transition rates \eqref{eq:lambda_def} and \eqref{eq:mu_def} were obtained by solving recursions similar to equations \eqref{eq:recursion_full} (see supplementary material).

The expected establishment time is an increasing function of the competition coefficient (dashed lines in Fig. \ref{fig:simplemodel} B and C). The weaker the competition, the higher the advantage of an initially rare introduced population and the lower the expected time to reach $K/2$. The exclusion time (dotted lines) on the other hand is a decreasing function of the competition coefficient. The stronger the competition, the weaker the force is that drives the system back to the coexistence point, and the shorter the time is to the exclusion of one of the two species. Due to these two opposing effects the total time (solid lines) to the extinction of the native species, twice the sum of establishment time and exclusion time, can exhibit a minimum. 

The higher the introduction rate, the smaller the influence of the competition coefficient is on establishment time, and the flatter the curve of $1/(2\phi(\alpha))$ will be. This is the reason why the position of the minimum is shifted to higher values of $\alpha$ as the introduction rate $\gamma$ increases (Fig. \ref{fig:simplemodel} C). At $\gamma=1$ the boundary where the introduced species is absent becomes an entrance boundary for the diffusion process; this means that the process can start at this  boundary but can never return to it \citep[][p. 235]{3238Karlin1981}. Thus, for higher introduction rates, establishment is no longer a limiting factor. To speed up the exclusion of the native species, the minimizing competition coefficient is above one in this region of parameter space. 

In Figs. \ref{fig:delta_effect} and \ref{fig:rgammadelta} we observed that in cases where the introduced and the native species differ in fecundity, the minimizing competition coefficient differs more from one than in the symmetric case. An intuitive explanation for this phenomenon is that in asymmetric cases the dynamics are strongly shaped by the differences in fecundity and large changes in the competition coefficient are required to affect these dynamics, whereas in the symmetric case small changes in the competition coefficient can tip the balance.

\begin{figure}
  \centering
  \subfloat[]{\label{fig:simplemodelstructure}\includegraphics[width=0.25\textwidth]{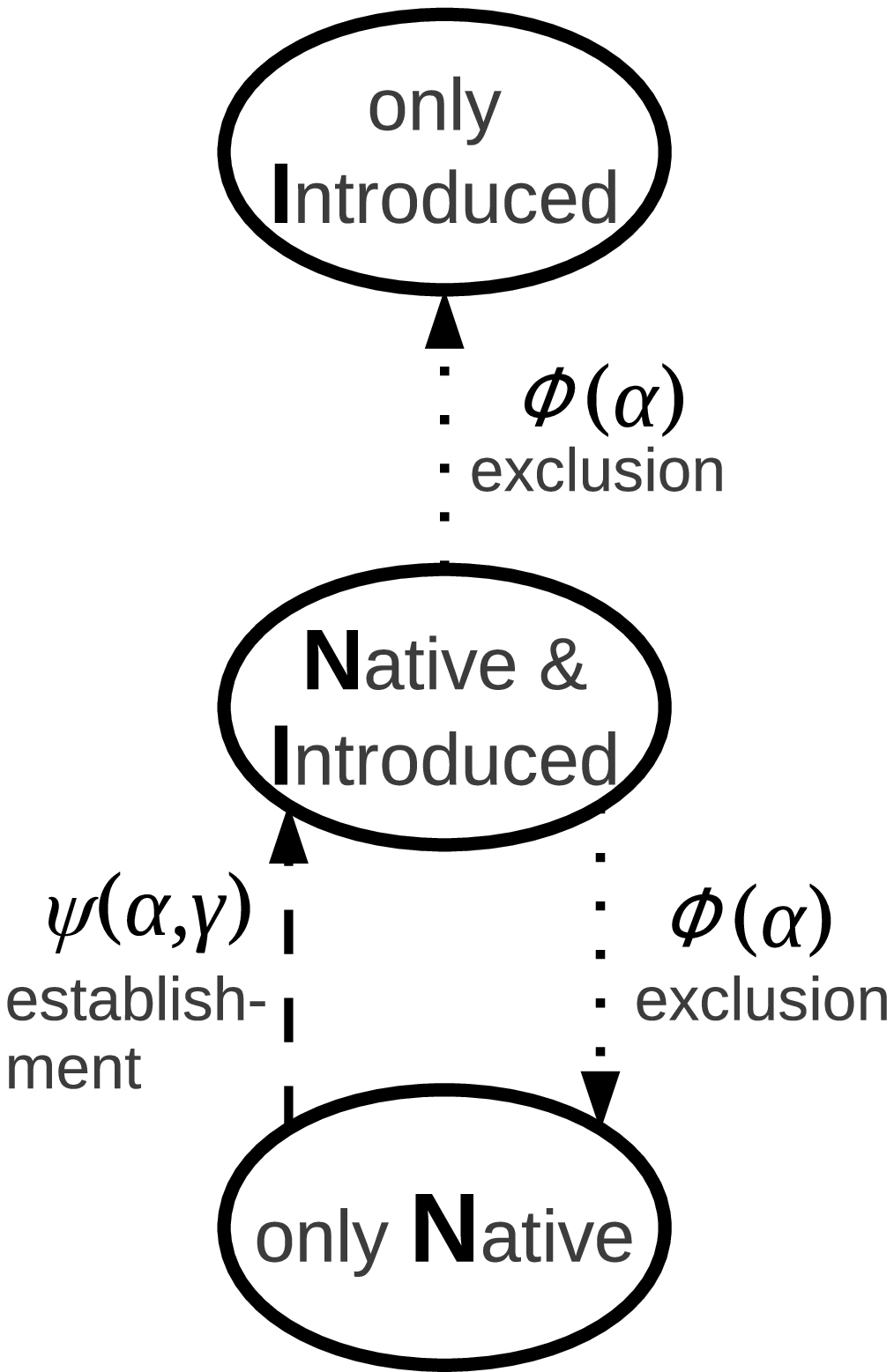}}                
  \subfloat[]{\label{fig:gammalow}\includegraphics[width=0.35\textwidth]{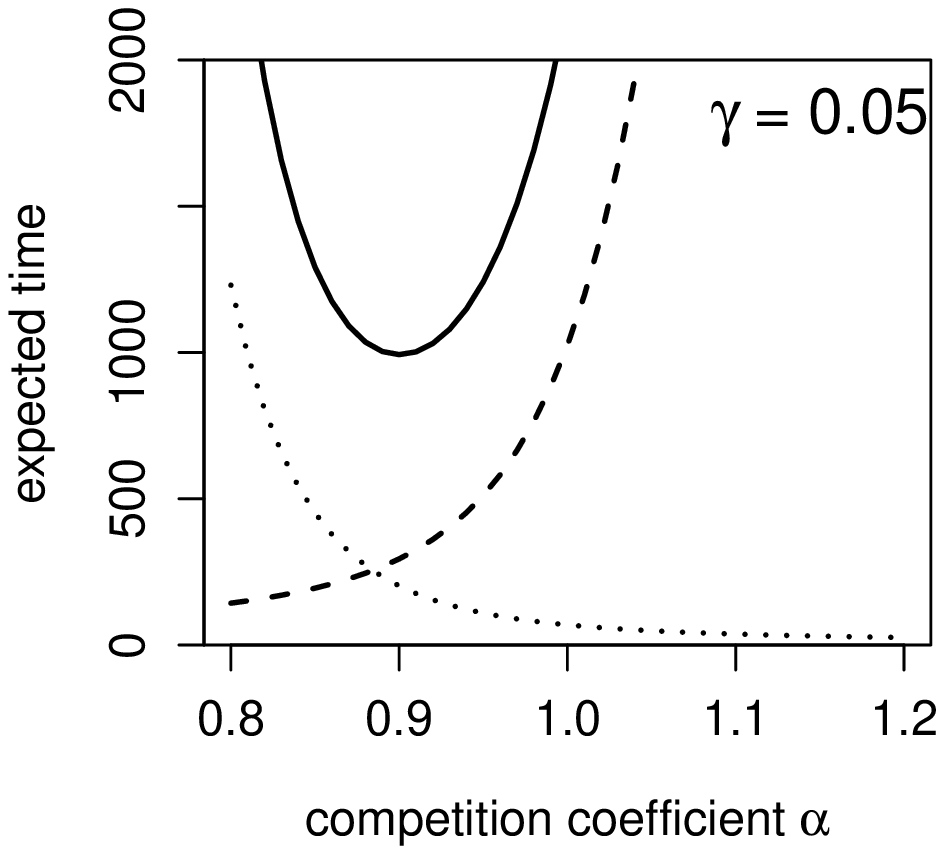}}
  \subfloat[]{\label{fig:gammahigh}\includegraphics[width=0.35\textwidth]{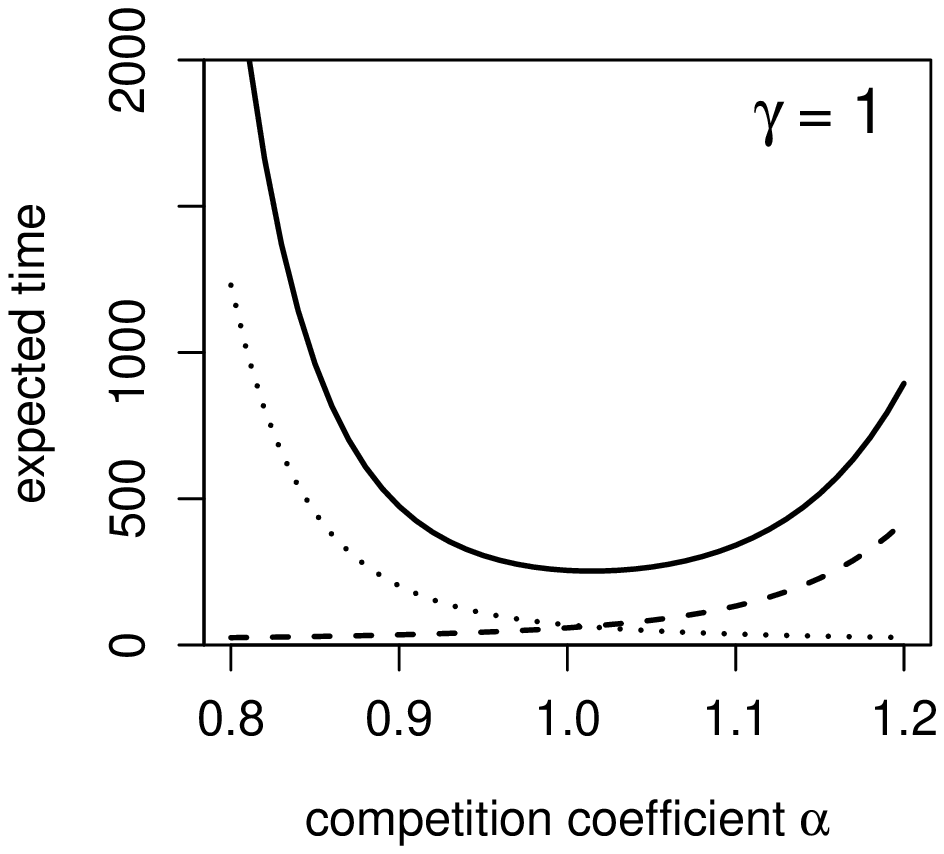}}
  \caption{A simple version of the model with only three states (A) illustrates the antagonistic effects of competition strength on establishment time (dashed lines) and exclusion time (dotted line) that lead to a minimum in the total time to native species extinction (solid line) which is at a smaller value of $\alpha$ in the case of a low introduction rate (B) and at a higher value for a high introduction rate (C). ($K=100, w=1$.)}
  \label{fig:simplemodel}
\end{figure}

To understand in which states the system spends most of its time, it is useful to examine the quasi-stationary distribution of the Markov process, the limiting distribution of population sizes of the native species given that it is not extinct yet. This is computed by eliminating the first row and first column, which belong to the absorbing state 0, from the rate matrix of the Markov process (see \eqref{eq:Lambda}).
The left eigenvector of the remaining matrix associated with the eigenvalue with the largest real part is the quasi-stationary distribution  \citep{3153DARROCH1967}. In the right column of Fig. \ref{fig:delta_effect} the quasi-stationary distribution is visualized for different parameter combinations. For competition coefficients on the right side of the minimum (indicated by squares), the introduced species is absent or has a low population size most of the time. For competition coefficients below the minimizing competition coefficient (triangles), the time the system spends around the coexistence point contributes most to the expected time to the extinction of the native species. 

In regions of the parameter space where the quasistationary distribution is very stable, the ratio of the expected time to extinction and its standard deviation is almost one, suggesting that the time to extinction is approximately exponentially distributed. Moving towards parameter combinations that lead to a fast extinction of the native species, for example as the introduced species' fecundity is increasing (see Fig. \ref{fig:varianceratio} B), the standard deviation is decreasing relative to the expectation, suggesting that the extinction of the native species becomes more deterministic.

The shift of the minimum in expected extinction time to smaller values of the competition coefficient in the model with genetic feedback can also be understood from the quasi-stationary distributions in Fig. \ref{fig:delta_effect}.
If competition is weak, the native population is most likely of intermediate size, and therefore a substantially lower amount of genetic diversity can be maintained within the population compared to a population that makes up the whole community. For high competition coefficients, the native population is most likely near its carrying capacity, if it is still present, which we condition on, and has therefore almost its full adaptability. 
Thus the additional extinction risk, or the reduction in expected time to extinction, would be larger for small competition coefficients than for large competition coefficients. This effect can be seen in Fig. \ref{fig:ecoevo}, as the curves for the eco-genetic model diverge more from the curve belonging to the ecological model at lower than they do at higher competition coefficients.

\section{Discussion}
Our theoretical results indicate that the introduction of nonnative competitors raises the extinction risk of native species, both directly and indirectly, via a reduction in genetic diversity. The expected impact does not generally increase or decrease with competition intensity as one might expect, but there is an intermediate competition coefficient for which the expected time to extinction of the native species is minimized. This is the result of the opposing effects of competition strength on the establishment step of the invasion process and on the impact of an already established species. Introduced species that do not compete intensively with species from their new range can readily establish, but their ecological impacts are weak and it will take a long time for them to drive one of the native competitors to extinction. On the other hand, an introduced species that is competing very intensely with one of the native species has high potential ecological impacts once it has established. However, such a species may need a lot of introduction attempts before it can establish, because the native competitor can efficiently exclude it from the community.

Based on our results we expect competitors with intermediate interaction strength to also have the lowest critical introduction rate. Thus if one would set a management target to preserve an endangered native species for a certain time period, then the greatest introduction prevention efforts would be necessary for nonnative species that would have an intermediate intensity of competition with the native species. 

With the help of our eco-genetic model, we quantified the feedback between ecological and genetic effects of the introduced species on the native competitor. This feedback is synergistic in the sense that ecological and genetic effects enhance each other: A reduction in population size causes a reduction in genetic diversity and this reduction in diversity can lead to further population decline in a changing environment. We found that this eco-genetic feedback is particularly strong for small intensities of competition between the introduced and the native species. This is because for high competition coefficients, the native species spends most of the time until its extinction in states with a high population size, whereas for intermediate and low competition coefficients the native species spends most of the time coexisting with the introduced species at an intermediate population size with a corresponding reduction in genetic diversity. Because the eco-genetic effect makes competitors with a relatively small competition intensity more dangerous for the native species, the minimizing competition coefficient is smaller compared to the ecological model. Similarly, also the critical introduction rate is lower than in the purely ecological scenario. This highlights the importance of including eco-genetic feedbacks into risk assessment models. If we only bring the introduction rate down to the level required under the ecological model and there is an eco-genetic feedback, we will fail to keep the expected impacts below the prescribed threshold. 

Eco-genetic effects are also a possible explanation for why there is relatively little evidence in the invasion biology literature for native species extinction due directly to an introduced competitor \citep{3305Davis2003}: Most endangered species are not threatened by a single stressor, but by combinations of them, for example habitat change and invasive species \citep{1138Gurevitch2004}. The introduction of a competitor can weaken a native species' ability to respond to other stressors. There may have been many extinction events which were attributed to other factors and in which a significant contribution from an introduced competitor went unrecognized. Thus, when making predictions on native species population dynamics it can be important to consider the possibly synergistic interaction of species invasions with other drivers of global change \citep{3276Didham2007}.

Previous theoretical studies on the impacts of introduced species have built models designed to understand these impacts in specific systems \citep{514Byers2001,48Thomson2005}. Here we contribute a building block towards the development of a theory that predicts impact from parameters of the introduced species, the native community and the introduction process. In this first stage of theory development our focus has been on simple models that are analytically tractable and give us insight into general phenomena. Of course, these models could be extended in many ways to incorporate more biological realism or to adjust them to specific biological systems, e.g. by including age or stage structure, which might in some cases influence the outcome of invasion and extinction dynamics \citep[see e.g.][]{360LANDE1988}. 

Our finding, however, that the native species' persistence time is minimized at intermediate intensities of competition, is robust to a wide range of model modifications. To illustrate this robustness, we examined a model in which the assumption of a fixed community size was relaxed, a model in which competition affects fecundity instead of viability, and a model with an alternative formulation of the transition rates similar to the one used in neutral community theory \citep[see e.g.][]{702Etienne2007}. All these models produced a minimum at intermediate competition intensities (see supplementary material for details of the analyses). Moreover, some of these modified models have the same diffusion approximation as the original model and thus behave very similarly, at least for large community sizes. Even a model in which we allow an immigration of native individuals from outside, exhibits a minimum in the expected time to the first extinction of the native species. However, possible measures of the long-term impact based on the stationary distribution of this process, like the proportion of time during which the native species is absent or the average native population size, have a monotonic relationship with competition intensity. This highlights that our results are most relevant for the short-term impacts of an introduced species on its native competitors.

In this study, we provide a model for the expected effect of an introduced species on one native competitor. Of course, native communities may consist of multiple competitors, as well as predators, mutualists, and parasites. To predict the impact of an introduced species on a whole community, our model could be combined with models for the other components of the community and interactions between them. Such detailed models have not been analyzed thus far.
However, Fig. 2 in a study by \citet{1157Case1990} shows that the probability that an introduced species can establish in a multi-species competitive community and replaces at least one native species is maximized for intermediate mean competition coefficients. Although \citet{1157Case1990} does not address this point, this is one hint that our finding that intermediate levels of competition are most dangerous scales up to more complex communities.

In our model, single individuals of the nonnative species were introduced into the new habitat. What happens if multiple nonnative individuals are released at once? \citet{1052Drake2005} found that the product of introduction frequency and introduction size was a good predictor for the persistence of introduced \textit{Daphnia} populations and that adding introduction frequency and size as single factors did not lead to significant improvements. However, if the nonnative population is subject to an Allee effect, i.e.\ positive density-dependence of population growth at low densities \citep{2079Courchamp1999}, the expected establishment success can strongly differ between a scenario with frequent introductions of one or a few individuals and one with rare introductions of many individuals \citep{314Drury2007}. Disentangling the effects of propagule size and propagule frequency for such important scenarios is a promising field of future research.

\section*{Acknowledgments}
We would like to thank Jonathan Jeschke and Joachim Hermisson for helpful discussions, as well as the handling editor Peter Chesson, two anonymous reviewers, and Sara Troxell for valuable comments on the manuscript. MJW is grateful to the Studienstiftung des deutschen Volkes for funding this project.

\renewcommand{\appendixpagename}{Appendix}
\begin{appendices}
\section{Rate matrix of the ecological model}
\label{sec:ratematrix}
\setcounter{equation}{0}
\renewcommand{\theequation}{A.\arabic{equation}}
The rate matrix of the ecological model is
\begin{equation}
\Lambda = \footnotesize
\begin{pmatrix}
-\lambda_0 & \lambda_0 & 0 & 0 & \cdots & 0 & 0 & 0\\
\mu_1 & -\left(\lambda_1 + \mu_1 \right) & \lambda _1 & 0 & \cdots & 0 & 0 & 0\\
0 & \mu_2 & -\left(\lambda_2 + \mu_2 \right) & \lambda _2 & \cdots & 0 & 0 & 0\\
\vdots & \vdots & \vdots & \vdots & \ddots & \vdots & \vdots & \vdots \\
0 & 0 & 0 & 0 & \cdots & \mu_{K-1} & -\left(\lambda_{K-1} + \mu_{K-1} \right) & \lambda _{K-1}\\
0 & 0 & 0 & 0 & \cdots & 0 & \mu_{K} & -\mu_{K}
\end{pmatrix} \normalsize \:,
\label{eq:Lambda}
\end{equation}
where the $\lambda_i$ and $\mu_i$ are given by equations \eqref{eq:lambda_def} and \eqref{eq:mu_def}.

\section{Full specification of the eco-genetic model}
\label{sec:ecogenetic_full}
\setcounter{equation}{0}
\renewcommand{\theequation}{B.\arabic{equation}}
Let $(n,m)$ be the state with $n$ native individuals, $m$ of which carry the currently favored allele. Then the transition rates for $1 \leq n \leq K$ and $0 \leq m \leq n$ are
\begin{equation}
(n,m) \rightarrow \begin{cases}
(n+1,m): & \frac{c(K-n,n)\cdot (K-n)}{K}\cdot \frac{(1-s)(n-m)(1-u)+mu}{\bar{w}_{n,m}} \;\text{ for } n  < K \\[0.5em]
(n-1,m): & \frac{c(n,K-n) \cdot (n-m)}{K}\cdot \frac{w^*\cdot w \cdot (K-n)}{\bar{w}_{n,m}} + \gamma \cdot \frac{c(n,K-n+1)\cdot(n-m)}{c(n,K-n+1)\cdot n + c(K-n+1,n)\cdot(K-n+1)}\; \text{ for } n > m\\[0.5em]
(n,m+1): & \frac{c(n,K-n)\cdot (n-m)}{K}\cdot \frac{m \cdot (1-u)+(1-s)(n-m)u}{\bar{w}_{n,m}}\; \text{ for } n > m\\[0.5em]
(n,m-1): & \frac{c(n,K-n)\cdot m}{K}\cdot \frac{(1-s)(n-m)(1-u)+mu}{\bar{w}_{n,m}} \; \text{ for } m > 0\\[0.5em]
(n-1,m-1): & \frac{c(n,K-n)\cdot m}{K}\cdot \frac{w^*\cdot w \cdot (K-n)}{\bar{w}_{n,m}}+ \gamma \cdot \frac{c(n,K-n+1)\cdot m}{c(n,K-n+1)\cdot n + c(K-n+1,n)\cdot(K-n+1)}\; \text{ for } m > 0\\[0.5em]
(n+1,m+1): & \frac{c(K-n,n)\cdot (K-n)}{K}\cdot \frac{m \cdot (1-u)+(1-s)(n-m)u}{\bar{w}_{n,m}} \text{ for }n < K\\[0.5em]
(n,n-m): & \epsilon \\[0.5em]
\end{cases}\:,
\end{equation}
where $\bar{w}_{n,m}=m+(1-s)(n-m)+w^*\cdot w \cdot (K-n)$ is the total fecundity in the community.
To obtain $w^*$, the average fecundity of the native species under the eco-genetic model in the absence of the introduced species ($n=K$), we computed the stationary distribution $\mathbf{d}=(d(0),\dots,d(K))$ of the birth and death process that describes the number of favored alleles in a native population of size $K$ and has transition rates  
\begin{equation}
m \rightarrow \begin{cases}
m+1: & (K-m) \cdot \frac{m \cdot (1-u)+(1-s)(K-m)u}{m+(1-s)(K-m)} \;\text{ for } m < K\\[0.5em]
m-1: & m \cdot \frac{mu+(1-s)(K-m)(1-u)}{m+(1-s)(K-m)}\; \text{ for } m > 0\\[0.5em] 
K-m: & \epsilon \; \text{ for }0\leq m \leq K\\[0.5em]
\end{cases}
\end{equation}
and averaged
\begin{equation}
w^* = 1-\frac{\sum_{m=0}^K d(m) (K-m) \cdot s}{K} \:.
\label{eq:w*}
\end{equation}

\section{Recursive solution for expectation and variance of the time to extinction in the ecological model}
\label{sec:recursionsolution}
\setcounter{equation}{0}
\renewcommand{\theequation}{C.\arabic{equation}}
The middle equation in the recursion \eqref{eq:recursion_full} can be rewritten as
\begin{equation}
\mu_n \left(\tau_n - \tau_{n-1}\right)=1+\lambda_n \left(\tau_{n+1} - \tau_n \right)\:.
\label{eq:recursion_rewrite}
\end{equation}
Define $z_n:=\tau_n - \tau_{n-1}$, such that
\begin{equation}
z_n = \frac{\lambda_n}{\mu_n}z_{n+1} + \frac{1}{\mu_n}\:.
\label{eq:z_recursion}
\end{equation}
Solving this recursion for $z$ with $z_K = \tau_K - \tau_{K-1} = \frac{1}{\mu_K}$ gives:
\begin{equation}
z_m = \sum_{i=m}^{K}\frac{1}{\mu_i} \prod_{j=m}^{i-1} \frac{\lambda_j}{\mu_j}\;\text{and}
\label{eq:zm}
\end{equation}

\begin{equation}
\tau_m = \sum_{l=1}^m z_l = \sum_{l=1}^m \sum_{i=l}^{K}\frac{1}{\mu_i} \prod_{j=l}^{i-1} \frac{\lambda_j}{\mu_j}\:.
\label{eq:Tm}
\end{equation}
With $m=K$ we obtain the expected persistence time \eqref{eq:T_K_solution}.

Similarly, the middle equation in \eqref{eq:recursion_variance} can be written as
\begin{equation}
\xi_n = \eta_n + \frac{\lambda_n}{\mu_n} \xi_{n+1}\:,
\end{equation}
where $\xi_n = \sigma_n^2-\sigma_{n-1}^2$ and $\eta_n$ is given by \eqref{eq:etas}.
As above, we can solve this recursion for $\xi$ with $\xi_K=\eta_K$:
\begin{equation}
\xi_l = \sum_{i=l}^{K} \eta_i \prod_{j=l}^{i-1} \frac{\lambda_j}{\mu_j}
\end{equation}
 and finally obtain \eqref{eq:variancesolution} by summing
\begin{equation}
\sigma_K^2= \sum_{l=1}^{K} \xi_l\:.
\end{equation}
\end{appendices}

\putbib
\end{bibunit}

\pagebreak

\renewcommand{\appendixpagename}{Supplementary material}
\begin{bibunit}
\begin{appendices}
\renewcommand{\figurename}{Fig.}
\renewcommand{\theequation}{S.\arabic{equation}}
\renewcommand{\thefigure}{S.\arabic{figure}}
\setcounter{equation}{0}
\setcounter{figure}{0}
\setcounter{section}{0}
\section*{Diffusion approximation to the ecological model}
We are seeking an approximation of equation (\ref{eq:T_K_solution}) for large community sizes $K$. The strategy used here is to first rescale the Markov process defined by the transition rates (\ref{eq:lambda_def}) and (\ref{eq:mu_def}) and approximate it by a diffusion process. Then we compute the expected time to native extinction in the diffusion process and finally prove that the exact recursive solution given by equation (\ref{eq:T_K_solution}) converges to the expected time under the diffusion process as $K$ goes to infinity.
\subsection*{Convergence in distribution $X_K \Rightarrow X$}
\begin{theorem}
As the community size $K$ goes to infinity, the rescaled Markov process
\begin{equation}
X_K = \left(X_K(t)\right)_{t\geq0}=\left(\frac{N(K\cdot t)}{K}\right)_{t \geq 0}.
\label{eq:rescaledprocess}
\end{equation}
with $\alpha=1+\frac{\beta}{K}$ and $w=1+\frac{\delta}{K}$, where $\beta$ and $\delta$ are constants, converges in distribution to a diffusion process $X$ with infinitesimal generator
\begin{linenomath}\begin{equation}
Lf(x)=  \frac{1}{2}b(x)\frac{d^2}{dx^2}f(x)+a(x)\frac{d}{dx}f(x)\:,
\label{eq:diffgenerator}
\end{equation}\end{linenomath}
where
\begin{linenomath}\begin{equation}
a(x) = -\beta (1-2x)(1-x)x - \delta x (1-x) - \gamma x
\end{equation}\end{linenomath}
is the infinitesimal mean of the diffusion process and 
\begin{linenomath}\begin{equation}
b(x)= 2x(1-x)
\end{equation}\end{linenomath}
is the infinitesimal variance.
\label{th:weakconvergence}
\end{theorem}
To prove the theorem we need a few lemmata.
\begin{lemma}
For $x \in \{0,\frac{1}{K},\frac{2}{K},\dots,1\}$ the generator of the Markov process $X_K$, defined as
\begin{linenomath}\begin{equation}
L_K f(x)  :=\frac{d}{dt} \mathbf{E}\big[f\big(X_K\left(t\right)\big)|X_K(0)=x\big]\big{|}_{t=0}\:,
\end{equation}\end{linenomath}
where $f$ is a bounded twice continuously differentiable function, converges to the generator of the diffusion process given by equation (\ref{eq:diffgenerator}):
\begin{equation}
\lim_{K \rightarrow \infty} \: \max_{x \in \{\frac{0}{K},\frac{1}{K},\dots,\frac{K}{K}\}} \: \left|L_K f(x) -Lf(x)\right|=0\:.
\label{eq:generatorconvergence}
\end{equation}
\label{th:generatorconvergence}
\end{lemma}
\begin{proof}
\begin{linenomath}\begin{equation}
L_K f(x) = \left[f\left(x+\frac{1}{K}\right)-f(x)\right]\cdot K \cdot \lambda_{Kx} + \left[f\left(x-\frac{1}{K}\right)-f(x)\right]\cdot K \cdot \mu_{Kx}
\end{equation}\end{linenomath}
\begin{linenomath}\begin{equation}
= \left[f'(x)+ \frac{1}{2}f''(x)\cdot \frac{1}{K} +\mathcal{O}\left(\frac{1}{K^2}\right)\right]\cdot \lambda_{Kx} + \left[-f'(x)+ \frac{1}{2}f''(x)\cdot \frac{1}{K} + \mathcal{O}\left(\frac{1}{K^2} \right)\right]\cdot \mu_{Kx}
\end{equation}\end{linenomath}

\begin{linenomath}\begin{equation}
= f'(x)\cdot \left(\lambda_{Kx} - \mu_{Kx}\right) + \frac{1}{2}f''(x)\cdot \frac{\lambda_{Kx} + \mu_{Kx}}{K}+ \mathcal{O}\left(\frac{1}{K^2} \right)\cdot \left(\lambda_{Kx} + \mu_{Kx}\right)\:.
\end{equation}\end{linenomath}
Substituting the scaled parameters into equations (\ref{eq:lambda_def}) and (\ref{eq:mu_def}) and using $o_i(1)$ to denote terms which fulfill $\lim\limits_{K \to \infty} o_i(1)=0$ uniformly in $x \in [0,1]$, we obtain
\begin{linenomath}\begin{align}
\lambda_{Kx} - \mu_{Kx} = & K \cdot \frac{\left(1+\frac{\beta}{K}x\right)(1-x)x-\left(1+\frac{\beta}{K}(1-x)\right)x(1-x)\left(1+\frac{\delta}{K}\right)}{1+\frac{\delta}{K}(1-x)} \\ \nonumber
 & -\gamma \cdot \frac{\left( 1+\frac{\beta}{K}(1-x)\right)x+o_1(1)}{1+2 x (1-x)\frac{\beta}{K} + o_2(1)}
\end{align}\end{linenomath}

\begin{linenomath}\begin{equation}
= \frac{-\beta (1-2x)(1-x) x - \delta(1-x)x +o_3(1)}{1+o_4(1)} - \gamma \cdot \frac{x+o_5(1)}{1+o_6(1)}=\mathcal{O}(1)\:,
\label{eq:lambdaKx-muKx}
\end{equation}\end{linenomath}
and
\begin{linenomath}\begin{equation}
\lambda_{Kx} + \mu_{Kx}=K \cdot \frac{2(1-x)x +o_7(1)}{1+o_4(1)} + \gamma \cdot \frac{x+o_5(1)}{1+o_6(1)} = \mathcal{O}(K)\:.
\label{eq:lambdaKx+muKx}
\end{equation}\end{linenomath}

Thus 
\begin{linenomath}\begin{equation}
L_K f(x) = x(1-x)\frac{d^2}{dx^2}f(x) + \big(-\beta (1-2x)(1-x)x - \gamma x - \delta x (1-x)\big)\frac{d}{dx}f(x) + o_8(1)\:.
\end{equation}\end{linenomath}
All expressions are bounded uniformly in $\{\frac{0}{K},\frac{1}{K},\dots,\frac{K}{K}\}$ and the error in equation (\ref{eq:generatorconvergence}) thus converges to zero uniformly in $\{\frac{0}{K},\frac{1}{K},\dots,\frac{K}{K}\}$ as $K$ goes to infinity.
\end{proof}

\begin{lemma}
The sequence of Markov processes $\left(X_K\right)_{K \in \mathbb{N}}$ is tight.
\label{th:tightness}
\end{lemma}
\begin{proof}
This is proven by using the compactness of the state space and the basic criterion for tightness and the Aldous condition from p. 34-35 in \citet{3351Joffe1986}.
\end{proof}

\begin{lemma}
The martingale problem for the generator $L$ has at most one solution.
\label{th:uniqueness}
\end{lemma}
\begin{proof}
The problem can be written as a stochastic differential equation with infinitesimal parameters extended beyond the interval $[0,1]$ with $a(x)=a(0)$ and $b(x)=b(0)$ for $x<0$ and $a(x)=a(1)$ and $b(x)=b(1)$ for $x>1$ such that the process is now defined on $\mathbb{R}$ while leaving the behavior inside the interval $[0,1]$ unchanged. Then we can apply the Yamada-Watanabe theorem (Theorem 26.10 in \citealt{3349Klenke2006}) and conclude that the stochastic differential equation has a unique strong solution. This implies uniqueness in law (Theorem 26.18 in \citealt{3349Klenke2006}), which in turn is equivalent to the statement that there exists at most one solution to the corresponding martingale problem \citep[][p. 159]{2250Rogers2000}.
\end{proof}

\begin{proof}[Proof of Theorem \ref{th:weakconvergence}]
The convergence in distribution follows from lemmata \ref{th:generatorconvergence}, \ref{th:tightness}, and \ref{th:uniqueness} and Theorem 4.8.10 in \citet{974Ethier2005}. 
\end{proof}

\subsection*{Expected time to extinction under the diffusion process}
Now we compute the expected time to the extinction of the native species under the diffusion process. Define
\begin{equation}
\sigma_l(X):= \inf\{t\geq 0 : X(t) \leq \l\}
\end{equation}
and let $\mathbf{E}_x$ denote the expectation for a process starting in $x$.

\begin{theorem}
\begin{equation}
\mathbf{E}_x[\sigma_0(X)]=\int_0^x e^{\beta \xi (1-\xi)-\gamma \ln(1-\xi)+\delta \xi} \int_{\xi}^1\frac{e^{-\beta \eta (1-\eta)+\gamma \ln(1-\eta)-\delta \eta}}{\eta (1-\eta)}d\eta \: d\xi\:.
\label{eq:g(1)}
\end{equation}
for $x \in [0,1]$.
\label{th:diffusiontime}
\end{theorem}
\begin{proof}
An important tool for the analysis of diffusion processes is the scale function $S(x)$. The scale function transforms the state space such that the diffusion process becomes a martingale. Thus the scale function fulfills the differential equation $LS(x) = 0$ \cite[][p.\ 196]{3238Karlin1981}. With $s(x):=\frac{d}{dx}S(x)$ the differential equation becomes
\begin{equation}
\frac{1}{2}b(x)\frac{d}{dx}s(x)+a(x)s(x) = 0 \:.
\end{equation}
Consequently,
\begin{equation}
s(x) = e^{\int_0^x -\frac{2 a(y)}{b(y)} dy}=e^{\int_0^x -\left[ -\beta (1-2y) - \frac{\gamma}{1-y} - \delta \right] dy}=e^{\beta x (1-x) - \gamma \ln(1-x) + \delta x}\:,
\end{equation}
where the lower limit of the integral can be chosen arbitrarily and is here 0 for convenience \cite[][p.\ 194]{3238Karlin1981}.
With this we can write the infinitesimal generator as \citep[][p.\ 195]{3238Karlin1981}:
\begin{equation}
Lf(x)=\frac{1}{2}s(x)b(x)\frac{d}{dx}\left[\frac{\frac{d}{dx}f(x)}{s(x)}\right]\:.
\label{eq:operator_modern}
\end{equation}

Let $g(x)$ be the solution of the differential equation \citep[][p. 193]{3238Karlin1981}
\begin{equation}
Lg(x)=-1
\label{eq:differentialequation2}
\end{equation}
with boundary conditions $g(0)=0$ and $|\lim\limits_{x \nearrow 1} g'(x)|<\infty$.
Using the differential operator from (\ref{eq:operator_modern}), equation (\ref{eq:differentialequation2}) is equivalent to
\begin{equation}
g(x)=\int_0^x e^{\beta \xi (1-\xi)-\gamma \ln(1-\xi)+\delta \xi} \left(C_1 - \int_{0.5}^\xi \frac{e^{-\beta \eta (1-\eta)+\gamma \ln(1-\eta)-\delta \eta}}{\eta (1-\eta)}d\eta \right) d\xi + C_2 \:.
\label{eq:general_solution}
\end{equation}
The boundary conditions are fulfilled with $C_2=0$ and  
\begin{equation}
C_1 = \int_{0.5}^1\frac{e^{-\beta \eta (1-\eta)+\gamma \ln(1-\eta)-\delta \eta}}{\eta (1-\eta)}d\eta \:.
\end{equation}
With this
\begin{equation}
g(x)=\int_0^x e^{\beta \xi (1-\xi)-\gamma \ln(1-\xi)+\delta \xi} \int_{\xi}^1\frac{e^{-\beta \eta (1-\eta)+\gamma \ln(1-\eta)-\delta \eta}}{\eta (1-\eta)}d\eta \: d\xi\:.
\end{equation}
Since $g(x)\leq g(1)<\infty$ for all $x \in [0,1]$, $g$ is a bounded, twice continuously differentiable function. Then, since $(X(t))_{t\geq0}$ solves the martingale problem for the generator $L$
\begin{equation}
g(X(t))-\int_0^t Lg(X(s)) ds = g(X(t)) + t
\end{equation}
with $t \in [0,\infty)$ is a martingale. Application of the optional stopping theorem \citep[see][p. 421, for a similar application]{974Ethier2005} with $\tau>0$ gives
\begin{equation}
\mathbf{E}_x\left[g\left(X(\sigma_0(X) \wedge \tau)\right)\right]+\mathbf{E}_x\left[\sigma_0(X) \wedge \tau \right]=g(x)\:.
\end{equation}
Letting $\tau$ go to infinity, we obtain by the dominated and monotone convergence theorems
\begin{equation}
g(x)=\mathbf{E}_x\left[\sigma_0(X) \right]\:.
\end{equation}
\end{proof}

\subsection*{Convergence of extinction times}
Because the first time to reach the boundary 0 is not a continuous functional of the process, we cannot conclude from the weak convergence of the rescaled birth and death process to the diffusion process that also the expected time to extinction converges to that under the diffusion process. Some more work is required to prove 
\begin{theorem}
\begin{equation}
\lim_{K \to \infty} \mathbf{E}_1[\sigma_0(X_K)] = \mathbf{E}_1[\sigma_0(X)]\:.
\end{equation}
\label{th:timeconvergence}
\end{theorem}

We start with proving the following lemmata:

\begin{lemma}
For all $j \in \{1,\dots,K-1\}$
\begin{equation}
\frac{\mu_j}{\lambda_j}=\frac{1+\frac{\beta}{K}\left(1-\frac{j}{K}\right)}{1+\frac{\beta}{K}\frac{j}{K}} \cdot \left[1+\frac{\delta}{K}+\frac{\gamma}{K-j}+\mathcal{O}\left(\frac{1}{K}\right)\right]\:.
\label{eq:mu/lambda}
\end{equation}
\label{th:quotient}
\end{lemma}
\begin{proof}
Substituting the scaled parameters given in theorem \ref{th:weakconvergence} into equations (\ref{eq:lambda_def}) and (\ref{eq:mu_def}), we find that
\begin{equation}
\lambda_j = \frac{\left(K+\frac{\beta}{K}\cdot j\right)\cdot (K-j)\cdot j}{K \cdot \left[K+\frac{\delta}{K}\cdot (K-j)\right]}
\label{eq:lambdajrescaled}
\end{equation}
and
\begin{align}
\label{eq:mujrescaled}
\mu_j = & \frac{\left[K+\frac{\beta}{K}\cdot(K-j)\right]\cdot j \cdot (K-j)\left(1+\frac{\delta}{K}\right)}{K \cdot \left[K+\frac{\delta}{K}\cdot (K-j)\right]} \\ \nonumber
 & + \gamma \cdot \frac{\left[K+1+\frac{\beta}{K}\cdot(K-j+1)\right]\cdot j}{\left[K+1+\frac{\beta}{K}\cdot(K-j+1)\right]\cdot j + \left[K+1+\frac{\beta}{K}\cdot j\right]\cdot(K-j+1)}\:.
\end{align}
From equations \ref{eq:lambdajrescaled} and \ref{eq:mujrescaled} it follows that
\begin{equation}
\frac{\mu_j}{\lambda_j}=\frac{\left[1+\frac{\beta}{K}\left(1-\frac{j}{K}\right)\right]\left(1+\frac{\delta}{K}\right)}{1+\frac{\beta}{K}\frac{j}{K}} + \frac{\gamma}{K-j} \frac{\left[1+\frac{1}{K}+\frac{\beta}{K}\left(1-\frac{j}{K}+\frac{1}{K}\right)\right] \left[1+\frac{\delta}{K} \left(1-\frac{j}{K}\right)\right]}{\left[\left(1+\frac{1}{K}\right)^2+2 \frac{\beta}{K}\left(1-\frac{j}{K}+\frac{1}{K}\right)\frac{j}{K}\right] \left(1+\frac{\beta}{K} \frac{j}{K}\right)}
\end{equation}
\begin{equation}
=\frac{\left[1+\frac{\beta}{K}\left(1-\frac{j}{K}\right)\right]\left(1+\frac{\delta}{K}\right)}{1+\frac{\beta}{K}\frac{j}{K}} + \frac{\gamma}{K-j} \frac{1+\frac{\beta}{K}\left(1-\frac{j}{K}\right)}{1+\frac{\beta}{K}\frac{j}{K}} \cdot \left(1+\mathcal{O}\left(\frac{1}{K}\right)\right)\:.
\end{equation}
After factorization one obtains (\ref{eq:mu/lambda}).
\end{proof}

\begin{lemma}
\begin{equation}
\prod_{j=l}^{i-1} \frac{\lambda_j}{\mu_j} < e^{2|\beta|+|\delta|} \cdot \frac{\left(1-\frac{i}{K}\right)^{\gamma}}{\left(1-\frac{l}{K}\right)^{\gamma}} \cdot e^{\mathcal{O}(1)} < e^{2|\beta|+|\delta|}\cdot e^{\mathcal{O}(1)}=:E
\label{eq:productbound}
\end{equation}
for all $l,i \in \{1,\dots,K\}$ with $l \leq i$.
\end{lemma}
\begin{proof}
Using lemma \ref{th:quotient} we can rewrite
\begin{equation}
\prod_{j=l}^{i-1} \frac{\lambda_j}{\mu_j} \nonumber
\end{equation}
as 
\begin{linenomath}\begin{equation}
\exp\left[-\sum_{j=l}^{i-1} \ln \left(\frac{1+\frac{\beta}{K}\left(1-\frac{j}{K}\right)}{1+\frac{\beta}{K}\frac{j}{K}} \cdot \left[1+\frac{\delta}{K}+\frac{\gamma}{K-j}+\mathcal{O}\left(\frac{1}{K}\right)\right]\right)\right]
\end{equation}\end{linenomath}
\begin{linenomath}\begin{equation}
=\exp\Bigg[\underbrace{\sum_{j=l}^{i-1}\ln\left(1+\frac{\beta}{K}\frac{j}{K}\right)- \ln\left(1+\frac{\beta}{K}\left(1-\frac{j}{K}\right)\right) - \ln\left(1+\frac{\delta}{K}+\frac{\gamma}{K-j}+\mathcal{O}\left(\frac{1}{K}\right) \right)}_F\Bigg]\:.
\end{equation}\end{linenomath}
Since $\ln(1+x)=x+\mathcal{O}(x^2)$ for $x$ close to zero and $-\ln(1+x)\leq -x + x^2$ for all $x \in (-0.5,\infty)$ 
\begin{equation}
F \leq \sum_{j=l}^{i-1} \left[\frac{\beta}{K}\frac{j}{K} -\frac{\beta}{K}\left(1-\frac{j}{K}\right)-\frac{\delta}{K}-\frac{\gamma}{K-j}+\frac{1}{(K-j)^2} + \mathcal{O}\left(\frac{1}{K}\right)\right]\:.
\end{equation}
and
\begin{equation}
\sum_{j=l}^{i-1} \frac{1}{(K-j)^2} < \sum_{j=1}^{\infty} \frac{1}{j^2} = \frac{\pi^2}{6} = \mathcal{O}\left(1\right)\:.
\end{equation}
With this
\begin{equation}
F \leq \sum_{j=l}^{i-1} \left[\frac{\beta}{K}\frac{j}{K} -\frac{\beta}{K}\left(1-\frac{j}{K}\right) -\frac{\delta}{K}-\frac{\gamma}{K-j}\right] + \mathcal{O}\left(1\right)
\end{equation}
\begin{equation}
\leq 2|\beta|+|\delta|-\gamma \ln\left(\frac{K-l}{K-i}\right)+ \mathcal{O}\left(1\right)
\end{equation}
and inequality (\ref{eq:productbound}) follows from this.
\end{proof}

\begin{lemma}
There exists a function $r(\psi)$ such that
\begin{equation}
\lim_{K \to \infty} \mathbf{E}_{\psi}[\sigma_0(X_K)] \leq r(\psi)
\end{equation}
and
\begin{equation}
\lim_{\psi \to 0} r(\psi) = 0\:.
\end{equation}
\label{th:T_K/K_convergence}
\end{lemma}

\begin{proof}
\begin{equation}
\mathbf{E}_{\psi}[\sigma_0(X_K)] \leq \frac{\tau_{\lceil K \psi \rceil}}{K} = \frac{1}{K}\sum_{l=1}^{\lceil K \psi \rceil} z_l = \frac{1}{K}\sum_{l=1}^m z_l + \frac{z_{\lceil K \psi \rceil}}{K}
\label{eq:Epsisum}
\end{equation}
with $m={\lceil K \psi \rceil}-1$ and where $\lceil x \rceil$ denotes the smallest integer larger than or equal to $x$. $\tau_i$ is given by equation (\ref{eq:Tm}) and $z_i$ by equation (\ref{eq:zm}). We first consider the first summand:
\begin{linenomath}\begin{equation}
\frac{1}{K}\sum_{l=1}^m z_l = \frac{1}{K}\sum_{l=1}^m \sum_{i=l}^{K}\frac{1}{\mu_i} \prod_{j=l}^{i-1} \frac{\lambda_j}{\mu_j}=\frac{1}{K}\sum_{i=1}^K \frac{1}{\mu_i} \sum_{l=1}^{\min(i,m)} \prod_{j=l}^{i-1} \frac{\lambda_j}{\mu_j}
\end{equation}\end{linenomath}

\begin{linenomath}\begin{equation}
= \frac{1}{K}\Bigg(\underbrace{\sum_{i=1}^m \frac{1}{\mu_i} \sum_{l=1}^{i} \prod_{j=l}^{i-1} \frac{\lambda_j}{\mu_j}}_A + \underbrace{\sum_{i=m+1}^{K-m} \frac{1}{\mu_i} \sum_{l=1}^{m} \prod_{j=l}^{i-1} \frac{\lambda_j}{\mu_j}}_B + \underbrace{\sum_{i=K-m+1}^{K-1} \frac{1}{\mu_i} \sum_{l=1}^{m} \prod_{j=l}^{i-1} \frac{\lambda_j}{\mu_j}}_C + \underbrace{\frac{1}{\mu_K} \sum_{l=1}^{m} \prod_{j=l}^{K-1} \frac{\lambda_j}{\mu_j}}_D\Bigg)\:.
\end{equation}\end{linenomath}
Note that
\begin{linenomath}\begin{equation}
\mu_i \geq K \cdot \frac{i}{K}\left(1-\frac{i}{K}\right) \frac{\left[1+\frac{\beta}{K}\left(1-\frac{i}{K}\right)\right]\left(1+\frac{\delta}{K}\right)}{1+\frac{\delta}{K}\left(1-\frac{i}{K}\right)}\geq K \cdot \frac{i}{K}\left(1-\frac{i}{K}\right) \underbrace{\frac{\left[1-\frac{|\beta|}{K}\right]\left(1-\frac{|\delta|}{K}\right)}{1+\frac{|\delta|}{K}}}_{=:H}\:.
\end{equation}\end{linenomath}
With this
\begin{linenomath}\begin{equation}
\frac{A}{K} \leq \frac{E}{K} \cdot \sum_{i=1}^m \frac{i}{\mu_i} \leq \frac{E}{HK} \sum_{i=1}^m \frac{1}{1-\frac{i}{K}} \leq \frac{E}{HK}\frac{m}{1-\frac{m}{K}}\leq \frac{E}{H}\frac{\psi}{1-\psi}\:,
\end{equation}\end{linenomath}
and
\begin{linenomath}\begin{equation}
\frac{B}{K} \leq \frac{m \cdot E}{HK} \sum_{i=m+1}^{K-m} \frac{K}{i(K-i)} \leq \frac{2m \cdot E}{HK} \cdot \sum_{i=m+1}^{K/2} \frac{K}{i\cdot \frac{K}{2}}
=\frac{4m \cdot E}{HK} \ln\left(\frac{K}{2m}\right) \leq \frac{4 \psi \cdot E}{H} \ln\left(\frac{1}{2\psi}\right)\:,
\end{equation}\end{linenomath}
and
\begin{equation}
\frac{C}{K} \leq \frac{E}{HK} \sum_{i=K-m+1}^{K-1} \frac{1}{i\left(1-\frac{i}{K}\right)} \sum_{l=1}^m \frac{\left(1-\frac{i}{K}\right)^{\gamma}}{\left(1-\frac{l}{K}\right)^{\gamma}}
\leq \frac{E \cdot m}{HK \cdot (K-m) \cdot \left(1-\frac{m}{K}\right)^{\gamma}} \sum_{i=K-m+1}^{K-1} \left(1-\frac{i}{K}\right)^{\gamma-1}\:.
\end{equation}
Let $f(x)=(1-x)^{\gamma-1}$. Then $f'(x)=(1-\gamma)(1-x)^{\gamma-2}$. $f'(x)<0$ if $\gamma>1$ and $f'(x)>0$ if $\gamma <1$ for all $x \in (0,1)$. 
If $f'(x)>0$
\begin{equation}
\sum_{i=K-m}^{K-1} f\left(\frac{i}{K}\right)
\end{equation}
is a lower sum of the integral
\begin{equation}
K \cdot \int_{1-\frac{m}{K}}^1 f(x) dx
\end{equation}
and if $f'(x)<0$
\begin{equation}
\sum_{i=K-m+1}^{K} f\left(\frac{i}{K}\right)
\end{equation}
is a lower sum. Thus, in both cases
\begin{equation}
\sum_{i=K-m+1}^{K-1} \left(1-\frac{i}{K}\right)^{\gamma-1} < K \int_{1-\frac{m}{K}}^1 (1-x)^{\gamma-1} dx = \frac{K}{\gamma}\left(\frac{m}{K}\right)^{\gamma}\:.
\end{equation}
Consequently
\begin{equation}
\frac{C}{K} \leq \frac{E \cdot m}{H \cdot (K-m) \cdot \left(1-\frac{m}{K}\right)^{\gamma}}\frac{1}{\gamma}\left(\frac{m}{K}\right)^{\gamma} \leq \frac{E\cdot \left(\psi\right)^{\gamma+1}}{H \cdot \gamma \cdot \left(1-\psi \right)^{\gamma+1}}\:.
\label{eq:Clim}
\end{equation}
Since
\begin{linenomath}\begin{equation}
\mu_K = \gamma \cdot \frac{\left(K+1+\frac{\beta}{K}\right)\cdot K}{\left(K+1+\frac{\beta}{K}\right)\cdot K+\left(K+1+\beta\right)}>\frac{\gamma}{2}\:,
\end{equation}
\begin{equation}
\frac{D}{K} \leq \frac{m \cdot E}{K \cdot \mu_K} \leq \frac{2E\psi}{\gamma} \:.
\end{equation}\end{linenomath}
Turning to the second summand in equation (\ref{eq:Epsisum}) and using equation (\ref{eq:zm})
\begin{linenomath}\begin{equation}
\frac{z_{\lceil K \psi \rceil}}{K}=\sum_{i=\lceil K \psi \rceil}^{K}\frac{1}{\mu_i} \prod_{j=\lceil K \psi \rceil}^{i-1} \frac{\lambda_j}{\mu_j}
\end{equation}
\begin{equation}
\leq \frac{E}{H}\sum_{i=\lceil K \psi \rceil}^{K-1}\frac{1}{i \cdot (K-i)} + \frac{E}{\mu_K \cdot K} \leq \frac{E}{H} \cdot \frac{\ln(K-\lceil K \psi \rceil)}{\lceil K \psi \rceil} + \frac{2E}{\gamma \cdot K} \xrightarrow[K \to \infty]{} 0\:.
\end{equation}\end{linenomath}
In conclusion
\begin{equation}
\lim_{K \to \infty} \mathbf{E}_{\psi}[\sigma_0(X_K)] \leq \frac{E}{H} \left( \frac{\psi}{1-\psi} + 4 \psi \ln\left(\frac{1}{2\psi}\right) + \frac{\psi^{\gamma+1}}{\gamma \cdot (1-\psi)^{\gamma+1}} \right) +\frac{2E\psi}{\gamma} =r(\psi)\:.
\end{equation}
To see that $\lim\limits_{\psi \to 0} r(\psi) = 0$ note that
\begin{equation}
\lim_{\psi \to 0} \psi \cdot \ln\left(\frac{1}{2\psi}\right) = \lim_{\psi \to 0} \frac{\frac{d}{d\psi} \ln\left(\frac{1}{2\psi}\right)}{\frac{d}{d\psi} \frac{1}{\psi}} = \lim_{\psi \to 0} \frac{\frac{1}{\psi}}{\frac{1}{\psi^2}}=\lim_{\psi \to 0} \psi =0\:.
\end{equation}
\end{proof}

\begin{proof}[Proof of Theorem \ref{th:timeconvergence}]
\begin{linenomath}\begin{equation}
\lim_{K \to \infty} \left|\mathbf{E}_1\left[\sigma_0(X_K)\right]-\mathbf{E}_1\left[\sigma_0(X)\right]\right|
\end{equation}
\begin{equation}
= \lim_{K \to \infty} \left|\mathbf{E}_1\left[\sigma_0(X_K)\right]-\mathbf{E}_1\left[\sigma_l(X_K)\right]+\mathbf{E}_1\left[\sigma_l(X_K)\right]-\mathbf{E}_1\left[\sigma_l(X)\right]+ \mathbf{E}_1\left[\sigma_l(X)\right] -\mathbf{E}_1\right[\sigma_0(X)\left]\right|
\end{equation}
\begin{equation}
\leq \lim_{K \to \infty} \bigg( \underbrace{\mathbf{E}_1\left[\sigma_0(X_K)\right]-\mathbf{E}_1\left[\sigma_l(X_K)\right]}_{=E_l[\sigma_0(X_K)]}+\left|\mathbf{E}_1\left[\sigma_l(X_K)\right]-\mathbf{E}_1\left[\sigma_l(X)\right]\right|+\left|\mathbf{E}_1\left[\sigma_l(X)\right] -\mathbf{E}_1\left[\sigma_0(X)\right]\right|\bigg)
\end{equation}
\begin{equation}
\leq r(l) + \lim_{K \to \infty} |\mathbf{E}_1[\sigma_l(X_K)]-\mathbf{E}_1[\sigma_l(X)]| + |\mathbf{E}_1[\sigma_l(X)]-\mathbf{E}_1[\sigma_0 (X)]|\:.
\end{equation}\end{linenomath}
provided that the limit in the second summand exists.
Following \citet{3347Kurtz1981} p. 13-14 we argue that for a fixed path $x$, $\sigma_l(x)$ is decreasing in $l$ and can therefore be discontinuous at only a countable number of values of $l$. Since the diffusion process is a rescaled Brownian motion, there is no point in the interval (0,1) that is exceptional compared to the other points and the points of discontinuity are placed according to some continuous probability distribution on the interval ($P(\text{discontinuity at } l)=0 \; \forall l$). For any discontinuity point thus the collection of paths for which the discontinuity is placed exactly at $l$ is a null set. Since the union of countably many null sets is a null set itself, the probability that a random path has a discontinuity at $l$ is zero.
Therefore, the weak convergence of $X_K$ to $X$ implies the convergence of $\sigma_l(X_K)$ to $\sigma_l(X)$ according to the Continuous Mapping Theorem (Theorem 13.25 in \citealt{3349Klenke2006}) such that the second limit exists and is equal to 0. See Lemma 3.3 in \citet{3348Chigansky2012} for a more rigorous proof requiring a positive and increasing scale function $S(x)$, which we can achieve by choosing
\begin{equation}
S(x)= \int_0^x e^{+\beta y (1-y) - \gamma \ln(1-y) + \delta y} dy\:
\end{equation}
and a positive quadratic variation $[X,X]_t$ for $t>0$.

After the middle summand vanished, we take the limit $l \to 0$:
\begin{equation}
\lim_{K \to \infty} \left|\mathbf{E}_1\left[\sigma_0(X_K)\right]-\mathbf{E}_1\left[\sigma_0(X)\right]\right| \leq \lim_{l \to 0} r(l) + \lim_{l \to 0}  |\mathbf{E}_1[\sigma_l(X)]-\mathbf{E}_1[\sigma_0 (X)]| = 0+0 = 0\:.
\end{equation}
The first limit follows from Lemma \ref{th:T_K/K_convergence}, the second limit from the boundedness of expected times under the diffusion process (Theorem \ref{th:diffusiontime}) and the application of dominated convergence. 
\end{proof}

Consequently, the expected time to extinction under the diffusion process is a valid approximation for the expected time under the rescaled birth and death process.
Figure \ref{fig:diffusionfits} visualizes the quality of the approximation for different community sizes.

\begin{figure}
\centering
\includegraphics[width=0.5\textwidth]{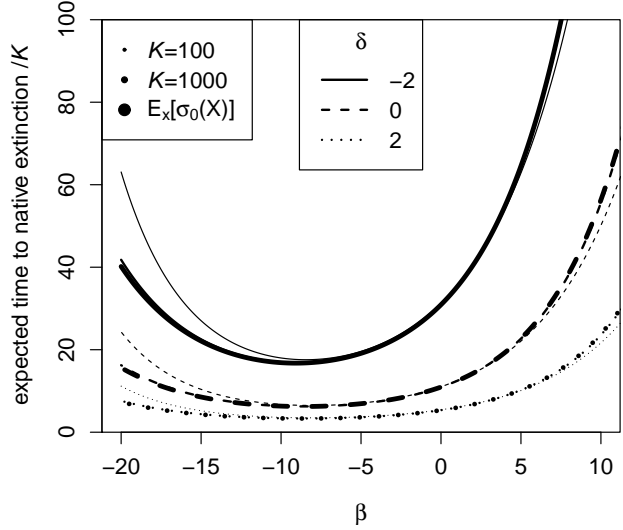}
\caption{The exact solution for the expected time to the extinction of the native species (\ref{eq:T_K_solution}) approaches the diffusion approximation (\ref{eq:g(1)}) as the community size increases. ($\gamma=0.1$).
\newline
\newline
}
\label{fig:diffusionfits}
\end{figure}

\section*{Derivation of establishment time and exclusion time}
The recursion for the expected time $\hat{\tau}_n$ to reach $K/2$ from a population size $n$ above $K/2$ is:
\begin{equation}
\hat{\tau}_n = \begin{cases}
0 & \text{if } n=\frac{K}{2} \\[0.5em]
\frac{1}{\lambda_n + \mu_n} + \frac{\lambda_n}{\lambda_n + \mu_n}\hat{\tau}_{n+1} + \frac{\mu_n}{\lambda_n + \mu_n}\hat{\tau}_{n-1} & \text{if }\frac{K}{2}+1\leq n \leq K-1 \\[0.5em]
\frac{1}{\mu_K}+\hat{\tau}_{K-1} & \text{if } n=K
\end{cases}\:.
\label{eq:recursion_establishment}
\end{equation}
This can be solved in the same way as recursion (\ref{eq:recursion_full}). For the expected establishment time if the native species starts from its carrying capacity, we obtain
\begin{equation}
\hat{\tau}_K = \sum_{l=K/2+1}^{K}\sum_{i=l}^{K}\frac{1}{\mu_i} \prod_{j=l}^{i-1} \frac{\lambda_j}{\mu_j}\:.
\label{eq:exact_esttime}
\end{equation}

We compute the exclusion time $\tilde{\tau}_{K/2}$ (the time to reach either 0 or $K$ from $K/2$) under the assumption of symmetry and set $\gamma=0$ in model equation \ref{eq:mu_def}. The appropriate recursion is:
\begin{equation}
\tilde{\tau}_n = \begin{cases}
0 & \text{if } n=0 \\[0.5em]
\frac{1}{\lambda_n + \mu_n} + \frac{\lambda_n}{\lambda_n + \mu_n} \tilde{\tau}_{n+1}+ \frac{\mu_n}{\lambda_n + \mu_n}\tilde{\tau}_{n-1} & \text{if }1\leq n \leq K-1 \\[0.5em]
0 & \text{if } n=K
\end{cases}\:.
\label{eq:recursion_exclusion}
\end{equation}
Equations (\ref{eq:recursion_rewrite}) and (\ref{eq:z_recursion}) remain valid but now we have $z_1 =  \tilde{\tau}_1 -  \tilde{\tau}_0 =  \tilde{\tau}_1$ and $z_K =  \tilde{\tau}_K -  \tilde{\tau}_{K-1} = -  \tilde{\tau}_{K-1}$, such that
\begin{equation}
z_m = -  \tilde{\tau}_{K-1} \prod_{j=m}^{K-1}\frac{\lambda_j}{\mu_j} + \sum_{i=m}^{K-1} \frac{1}{\mu_i} \prod_{j=m}^{i-1}\frac{\lambda_j}{\mu_j}\;\;\text{ and}
\end{equation}

\begin{equation}
 \tilde{\tau}_1 = z_1 = -  \tilde{\tau}_{K-1} \prod_{j=1}^{K-1}\frac{\lambda_j}{\mu_j} + \sum_{i=1}^{K-1} \frac{1}{\mu_i} \prod_{j=1}^{i-1}\frac{\lambda_j}{\mu_j}\:.
\end{equation}
Because of symmetry
\begin{equation}
 \tilde{\tau}_1 =  \tilde{\tau}_{K-1} = \frac{\sum_{i=1}^{K-1} \frac{1}{\mu_i} \prod_{j=1}^{i-1}\frac{\lambda_j}{\mu_j}}{1+\prod_{j=1}^{K-1}\frac{\lambda_j}{\mu_j}}=\frac{\sum_{i=1}^{K-1} \frac{1}{\mu_i} \prod_{j=1}^{i-1}\frac{\lambda_j}{\mu_j}}{2} \;\;\text{ and}
\end{equation}

\begin{equation}
 \tilde{\tau}_{K/2}= \sum_{l=1}^{K/2}\sum_{i=l}^{K-1} \frac{1}{\mu_i} \prod_{j=1}^{i-1}\frac{\lambda_j}{\mu_j} - \frac{1}{2}\left(\sum_{i=l}^{K-1} \frac{1}{\mu_i} \prod_{j=1}^{i-1}\frac{\lambda_j}{\mu_j}\right)\left(\sum_{l=1}^{K/2}\prod_{j=l}^{K-1}\frac{\lambda_j}{\mu_j}\right)\:.
\label{eq:exactfixtime}
\end{equation}

Under the diffusion approximation, expressions for establishment time $\hat{g}$ and exclusion time $\tilde{g}$ can be found by using appropriate boundary conditions in equation (\ref{eq:general_solution}). For the establishment time in the symmetric case, we need to use $\hat{g}(1/2)=0$ and $|\lim\limits_{x \nearrow 1} \hat{g}'(x)|<\infty$, such that 
\begin{equation}
\hat{g}(1)=\int_{0.5}^1 \frac{1}{(1-\xi)^{\gamma}} \int_{\xi}^1 \frac{(1-\eta)^{\gamma-1}}{\eta}\cdot e^{\beta\left[\xi(1-\xi)-\eta(1-\eta)\right]} d \eta \;d\xi \:.
\end{equation}
For the exclusion time, we need to use the boundary conditions $\tilde{g}(1)=\tilde{g}(0)=0$ and we set $\gamma=0$, which leads to
\begin{equation}
\tilde{g}(0.5)=\int_0^{1/2}\frac{e^{-\beta \eta (1-\eta)}}{\eta(1-\eta)} \int_0^{\eta} e^{\beta \xi (1-\xi)} d \xi \;d \eta\:.
\end{equation}

\section*{Robustness of results}
In this study, we have shown that according to our model for the competitive dynamics of a native and an introduced species, the expected time to the extinction of the native species is minimized for intermediate intensities of interspecific competition. In this section we explore to what extent this result is robust to model modifications.

\subsection*{The assumption of a fixed community size}
In our original model, we assumed a fixed total community size. This assumption induces a strong coupling between the population dynamics of the two species. Here we consider a model in which the population dynamics are coupled only because individual death rates are proportional to the competition experienced, which is a function of both population sizes. To describe the system, we then need both the population size $n_1$ of the native species and the population size $n_2$ of the introduced species. We thus define a Markov process with state space is $\{0,1,2,\dots\}\times\{0,1,2,\dots\}$ and transition rates
\begin{equation}
(n_1,n_2)\rightarrow \begin{cases}
(n_1+1,n_2): & n_1 \\
(n_1-1,n_2): & \frac{c(n_1,n_2)}{K} \cdot n_1 \\
(n_1,n_2+1): & w_i n_2 + \gamma \\
(n_1,n_2-1): & \frac{c(n_2,n_1)}{K} \cdot n_2
\end{cases}\:.
\label{eq:nonzerosumtransitions}
\end{equation}
Note that in this model, there is no strict upper limit to the population sizes. However, when the native  species is alone its death rate exceeds its birth rate for population sizes larger than $K$ and similarly for the introduced species when its size is larger than $w_i K$. When both species are together, the total community size fluctuates around a value which, if $\alpha<1$, can exceed both of these quantities and which increases with decreasing competition intensity. The behavior is thus similar to the classical Lotka-Volterra competition model.

Using the same approach as for the eco-genetic model, we solved numerically for the expected time to the extinction of the native species when its initial size is $K$ while the introduced species is initially absent. To obtain a finite transition matrix we needed to assume a maximum community size. We chose $4\cdot \max(Kw_i,K)$ as excursions of the system to such a high community size are unlikely. For all parameter combinations we examined, this approximation to the expected time to extinction exhibits a minimum at intermediate competition intensities (see Fig. \ref{fig:nonzerosum} for an example). Due to the two-dimensional state space, these computations are only possible for small $K$.

\begin{figure}
\centering
\includegraphics[width=0.5\textwidth,trim = 0mm 0mm 0mm 15mm, clip]{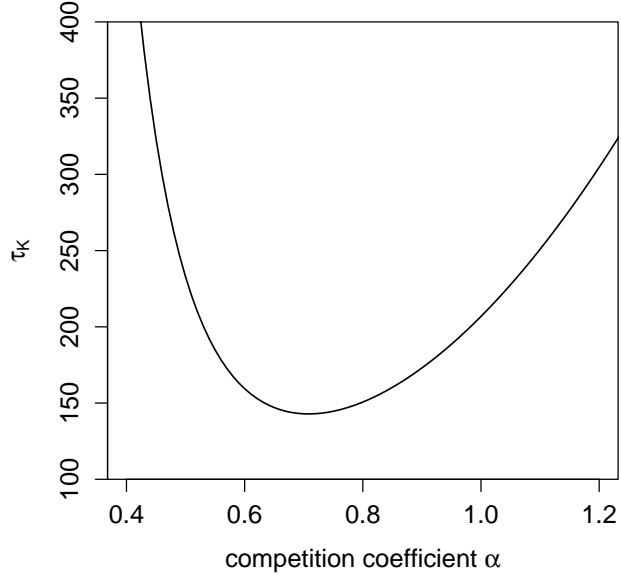}
\caption{The expected time to the extinction of the native species $\tau_K$ in the model without a fixed community size with transition rates given by \eqref{eq:nonzerosumtransitions}. ($K=20,\gamma=0.1,w_i=1$)}
\label{fig:nonzerosum}
\end{figure}

\subsection*{Formulation of transition rates as in neutral community theory}
Our model has some parallels to neutral community theory and related non-neutral models. However, most of these models have a different way of incorporating immigration. If a vacancy in the local community is created, with probability $1-m$ it is filled by the offspring of an individual that is already in the community and with probability $m$ by a migrant from the metacommunity \citep[see e.g.][]{702Etienne2007}.
In our case, immigrations from the metacommunity correspond to introduction events. In analogy to the transition rates in neutral community theory, we can reformulate our transition rates \eqref{eq:lambda_def} and \eqref{eq:mu_def} as:
\begin{linenomath}\begin{equation}
\lambda_n =\underbrace{\frac{c(K-n,n)\cdot (K-n)}{K}}_{\substack{\text{rate at which}\\ \text{members of the introduced}\\ \text{species die}}} \cdot \underbrace{(1-m)\cdot \frac{n}{(K-n)\cdot w+ n}}_{\substack{\text{probability that a}\\ \text{native individual} \\ \text{gives birth}}}
\label{eq:lambdaneutral_def}
\end{equation}\end{linenomath}
and
\begin{linenomath}\begin{equation}
\mu_n = \underbrace{\frac{c(n,K-n)\cdot n}{K}}_{\substack{\text{rate at which}\\ \text{native individuals die}}} \cdot \Big[ \underbrace{(1-m)\cdot \frac{(K-n)\cdot w}{(K-n) \cdot w+ n}}_{\substack{\text{probability that an}\\ \text{introduced individual} \\ \text{already in the community}  \\ \text{gives birth}}} + \underbrace{m}_{\substack{\text{probability that the} \\ \text{spot is colonized by a} \\ \text{new introduced individual}}}\Big]\:.
\label{eq:muneutral_def}
\end{equation}\end{linenomath}
We explored a range of parameter combinations and found that the expected time to the extinction of the native species under this model (Fig. \ref{fig:TK_neutral}) behaves very similarly to that under the original model (compare Fig. \ref{fig:omega0}). The minimum is preserved.

\begin{figure}
\centering
\includegraphics[width=0.5\textwidth,trim = 0mm 0mm 0mm 15mm, clip]{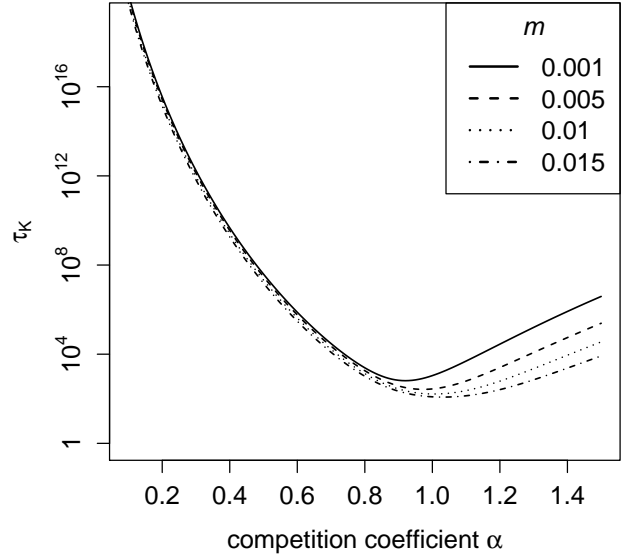}
\caption{The expected time to the extinction of the native species $\tau_K$ in a model with transition rates analogous to those in neutral community theory (see equations \ref{eq:lambdaneutral_def} and \ref{eq:muneutral_def}). ($K$ = 100, $w=1$.)}
\label{fig:TK_neutral}
\end{figure}

Moreover, if we assume the parameter scaling $\gamma=m/K$, then the Markov process converges to the diffusion process specified by the infinitesimal generator in \eqref{eq:generator}. To see this note that for $x \in [0,1]$
\begin{linenomath}\begin{align}
\lambda_{Kx} - \mu_{Kx} = & K \cdot \left(1-\frac{\gamma}{K}\right) \cdot \frac{\left(1+\frac{\beta}{K}x\right)(1-x)x-\left(1+\frac{\beta}{K}(1-x)\right)x(1-x)\left(1+\frac{\delta}{K}\right)}{1+\frac{\delta}{K}(1-x)} \\ \nonumber
 & -\frac{\gamma}{K} \cdot K \left(1+\frac{\beta}{K}(1-x)\right)x
\end{align}\end{linenomath}

\begin{linenomath}\begin{equation}
= \frac{-\beta (1-2x)(1-x) x - \delta(1-x)x +o_3(1)}{1+o_4(1)} \cdot (1+o_5(1)) - \gamma \cdot x+o_6(1)\:,
\end{equation}\end{linenomath}
and
\begin{linenomath}\begin{equation}
\lambda_{Kx} + \mu_{Kx}=K \cdot \frac{2(1-x)x +o_7(1)}{1+o_4(1)} \cdot (1+o_5(1)) + \gamma  \cdot x+o_6(1)
\end{equation}\end{linenomath}
similarly to \eqref{eq:lambdaKx-muKx} and \eqref{eq:lambdaKx+muKx}.  Thus the convergence of generators \eqref{eq:generatorconvergence} holds. The preconditions for the proofs of lemmata \ref{th:tightness} and \ref{th:uniqueness} are not affected by the modification of transition rates and thus by Theorem \ref{th:weakconvergence} the Markov process with transition rates \eqref{eq:lambdaneutral_def} and \eqref{eq:muneutral_def} converges in distribution to the same diffusion process as the original Markov process.

\subsection*{Competition affecting fecundity instead of mortality}
An alternative formulation of the transition rates \eqref{eq:lambda_def} and \eqref{eq:mu_def} in which competition affects the rate at which individuals give birth rather than the death rate, is:
\begin{linenomath}\begin{equation}
\lambda_n=\underbrace{\left(2-\frac{c(n,K-n)}{K}\right)\cdot n}_{\substack{\text{rate at which} \\ \text{native individuals} \\ \text{give birth}}} \; \cdot \underbrace{\frac{K-n}{K}}_{\substack{\text{probability that an} \\ \text{introduced individual dies}}}
\label{eq:lambdabirth}
\end{equation}\end{linenomath}
and
\begin{linenomath}\begin{equation}
\mu_n=\underbrace{w_i \cdot \left(2-\frac{c(K-n,n)}{K}\right)\cdot (K-n)}_{\substack{\text{rate at which} \\ \text{introduced individuals} \\ \text{give birth}}} \; \cdot \underbrace{\frac{n}{K}}_{\substack{\text{probability that a} \\ \text{native individual dies}}}+ \; \gamma \cdot \underbrace{\frac{n}{K+1}}_{\substack{\text{probability that the}\\ \text{introduced individual} \\ \text{replaces a native one}}}\:.
\label{eq:mubirth}
\end{equation}\end{linenomath}
The phenomenon that the expected time to native extinction is minimal at intermediate competition intensities is apparently not affected by this change in model formulation (Fig. \ref{fig:TK_birthcomp}). Again, evaluating $\lambda_{Kx} - \mu_{Kx}$ and $\lambda_{Kx} + \mu_{Kx}$ reveals that this Markov process converges to the same diffusion process as the original Markov process.

\begin{figure}
\centering
\includegraphics[width=0.5\textwidth,trim = 0mm 0mm 0mm 15mm, clip]{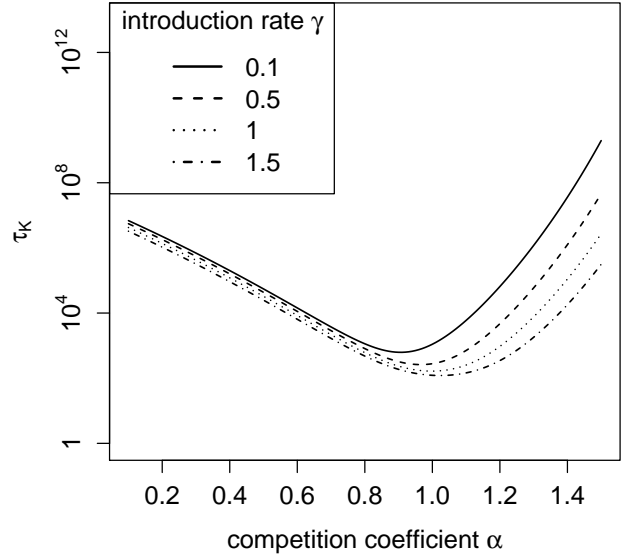}
\caption{The expected time to the extinction of the native species $\tau_K$ in a model in which competition influences the birth rate rather than the death rate as specified by \eqref{eq:lambdabirth} and \eqref{eq:mubirth}. ($K$ = 100, $w=1$.)}
\label{fig:TK_birthcomp}
\end{figure}

\subsection*{Immigration of native individuals}
To allow for an immigration of native individuals at rate $\gamma'$, we modify \eqref{eq:lambda_def} by adding an immigration term analogous to the introduction term in \eqref{eq:mu_def}:
\begin{linenomath}\begin{align}
\lambda_n = & \underbrace{\frac{c(K-n,n)\cdot (K-n)}{K}}_{\substack{\text{rate at which}\\ \text{members of the introduced}\\ \text{species die}}} \cdot \underbrace{\frac{n}{(K-n)\cdot w+ n}}_{\substack{\text{probability that a}\\ \text{native individual} \\ \text{gives birth}}} \nonumber \\
 & + \underbrace{\gamma'}_{\substack{\text{immigration} \\ \text{rate}}} \cdot \underbrace{\frac{c(K-n,n+1)\cdot (K-n)}{c(n+1,K-n)\cdot (n+1) + c(K-n,n+1)\cdot(K-n)}}_{\substack{\text{probability that an introduced}\\ \text{individual dies}}}\:.
\label{eq:lambdawithimmigration_def}
\end{align}\end{linenomath}
With this modification of transition rates, the state in which the native species is absent is no longer an absorbing state of the Markov process. Result \eqref{eq:T_K_solution} now gives us the expected time until the first extinction of the native species for realizations starting with $K$ individuals of the native species. Evaluating the expected time to the first extinction for a range of parameter combinations suggests that this quantity is higher than the expected time to extinction in the original model, but still exhibits a minimum at intermediate competition coefficients (Fig. \ref{fig:TK_withimmigration}).

\begin{figure}
\centering
\includegraphics[width=0.5\textwidth,trim = 0mm 0mm 0mm 15mm, clip]{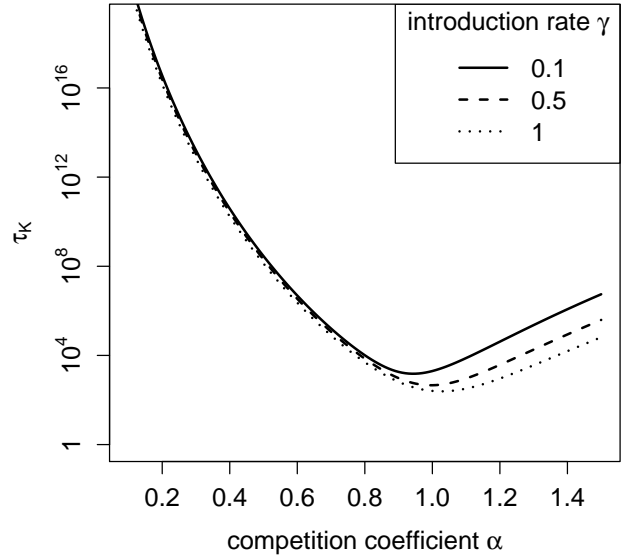}
\caption{The expected time to the first extinction $\tau_K$ of the native species if we allow for the immigration of native individuals at rate $\gamma'=0.5$ as specified in \eqref{eq:lambdawithimmigration_def}. ($K$ = 100, $w=1$.)}
\label{fig:TK_withimmigration}
\end{figure}

One might also be interested in long-term measures of impact based on the stationary distribution of the Markov process, for example the proportion of time during which the native species is absent or the average population size of the native species. Unlike the expected time to the first extinction, these quantities exhibit monotonic relationships with competition intensity (Fig. \ref{fig:longterm}).

\begin{figure}
\centering
\subfloat[]{\includegraphics[width=0.5\textwidth,trim = 0mm 0mm 0mm 15mm, clip]{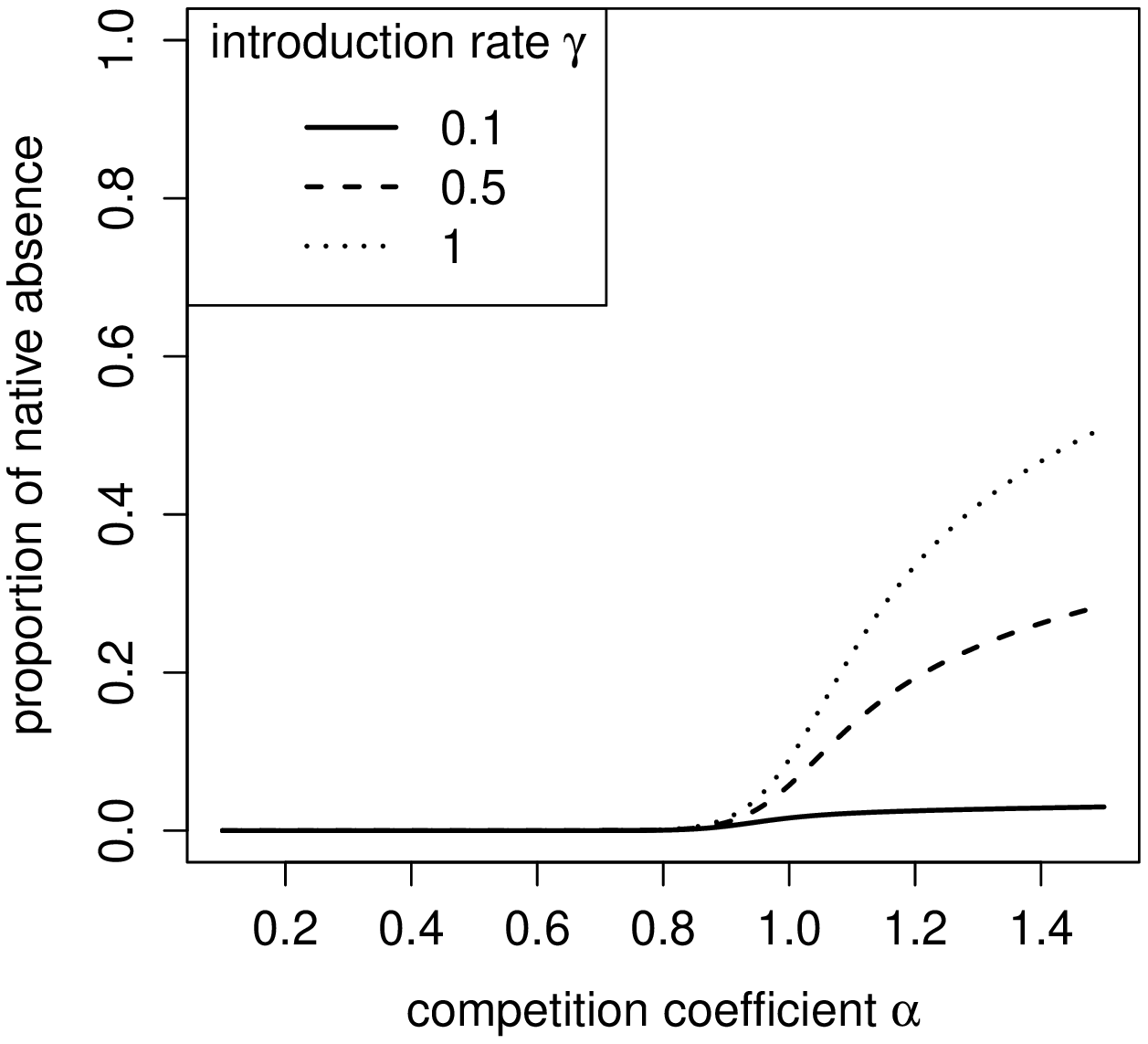}}
\subfloat[]{\includegraphics[width=0.5\textwidth,trim = 0mm 0mm 0mm 15mm, clip]{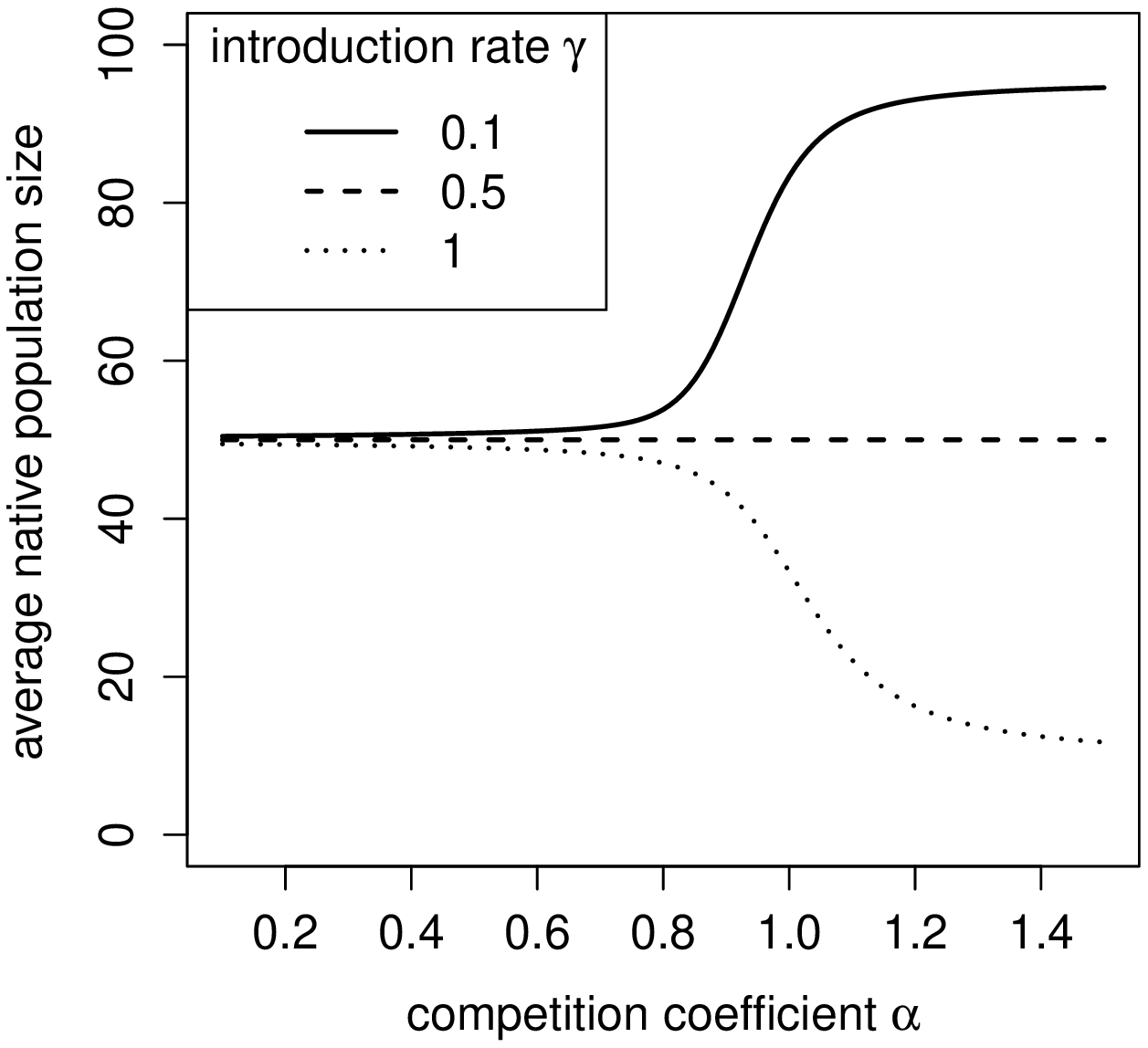}}
\caption{Long-term measures of impact in the model with immigration of native individuals: A) the proportion of time during which the native species is absent and B) the average population size of the native species under the stationary distribution. ($K$ = 100, $w=1$, $\gamma'=0.5$.)}
\label{fig:longterm}
\end{figure}
\end{appendices}

\pagebreak

\putbib
\end{bibunit}

\end{document}